\documentclass[reqno,a4paper,11pt]{amsart}
\usepackage{comment,amssymb,latexsym,upref,enumerate,fouridx}
\usepackage{dutchcal}
\usepackage{mathrsfs,xcolor}
\usepackage{graphicx}
\graphicspath{ {./images/} }
\usepackage{subfigure}
\usepackage{placeins}
\usepackage{centernot}
\usepackage[normalem]{ulem}
\usepackage[colorlinks,linkcolor=blue,citecolor=blue]{hyperref}

\allowdisplaybreaks
\numberwithin{equation}{section}

\theoremstyle{plain}
\newtheorem{thm}{Theorem}[section]
\newtheorem{lem}[thm]{Lemma}
\newtheorem{prop}[thm]{Proposition}
\newtheorem{cor}[thm]{Corollary}

\theoremstyle{definition}

\theoremstyle{remark}
\newtheorem{rem}[thm]{Remark}

\newcommand{\wttt}{\tilde{\wtt}{}}
\newcommand{\zttt}{\tilde{\ztt}{}}
\newcommand{\ttld}{\tilde{t}{}}
\newcommand{\Wtth}{\hat{\Wtt}{}}
\newcommand{\wtth}{\hat{\wtt}{}}
\newcommand{\ztth}{\hat{\ztt}{}}
\newcommand{\Wttt}{\tilde{\Wtt}{}}
\newcommand{\Attt}{\tilde{\Att}{}}
\newcommand{\Bttt}{\tilde{\Btt}{}}
\newcommand{\sspeed}{\sigma}

\newcommand{\pip}{\pi^{\perp}}

\newcommand{\wttb}{\bar{\wtt}{}}
\input Tdef.def

\begin{document}

\title[On the fractional density gradient blow-up conjecture of Rendall]{On the fractional density gradient blow-up conjecture of Rendall}

\author[T.A.~Oliynyk]{Todd A.~Oliynyk}
\address{School of Mathematics\\
9 Rainforest Walk\\
Monash University, VIC 3800\\ Australia}
\email{todd.oliynyk@monash.edu}

\begin{abstract} 
On exponentially expanding Friedmann-Lema\^{i}tre-Robertson-Walker (FLRW) spacetimes, there is a distinguished family of spatially homogeneous and isotropic solutions to the relativistic Euler equations with a linear equation of state of the form $p=\sigma \rho$, where $\sigma \in [0,1]$ is the square of the sound speed. Restricting these solutions to a constant time hypersurface yields initial data that uniquely generates them. In this article, 
we show, for sound speeds satisfying $\frac{1}{3}<\sigma <\frac{k+1}{3k}$ with $k\in \Zbb_{>\frac{3}{2}}$, that $\Tbb^2$-symmetric initial data that is chosen sufficiently close to spatially homogeneous and isotropic data uniquely generates a $\Tbb^2$-symmetric solution of the relativistic Euler equations that exists globally to the future. Moreover, provided $k\in \Zbb_{>\frac{5}{2}}$, we show that there exist open sets of $\Tbb^2$-symmetric initial data for which the fractional density gradient becomes unbounded at timelike infinity. This rigorously confirms, in the restricted setting of relativistic fluids on exponentially expanding FLRW spacetimes, the fractional density gradient blow-up scenario conjectured by Rendall in \cite{Rendall:2004}.
\end{abstract}

\maketitle

\section{Introduction\label{intro}}
Relativistic perfect fluids with a linear equation of state on a prescribed spacetime $(M,\gt)$ are governed by the relativistic Euler equations\footnote{See Section \ref{indexing} for our indexing conventions.} 
\begin{equation}
\nablat_i \Tt^{ij}=0 \label{relEulA}
\end{equation}
where 
\begin{equation*}
\Tt^{ij} = (\rho+p)\vt^i \vt^j + p \gt^{ij}
\end{equation*}
is the stress energy tensor, $\vt^{i}$ is the fluid
four-velocity normalized by $\gt_{ij}\vt^i \vb^j=-1$, and the fluid's proper energy density $\rho$ and pressure $p$ are related by
\begin{equation*} 
p = \sspeed \rho.
\end{equation*} 
Since $\sspeed=\frac{dp}{d\rho}$ is the square of the sound speed, we will always assume\footnote{While this restriction on the sound speed is often taken for granted, it is, strictly speaking, not necessary; see \cite{Geroch:2010} for an extended discussion.} 
that $0\leq \sspeed \leq 1$
in order to ensure that the speed of sound is less than or equal to the speed of light. In this article, we will restrict our attention to analysing the relativistic Euler equations on
exponentially expanding Friedmann-Lema\^{i}tre-Robertson-Walker (FLRW) spacetimes $(M,\gt)$ where
\begin{equation}\label{M-def} 
M = (0,1]\times \Tbb^3,
\end{equation}
and\footnote{By introducing a change of time coordinate via $\tilde{t}=-\ln(t)$, the metric \eqref{FLRW} can be brought into the more recognizable form $\gt = -d\tilde{t}\otimes d\tilde{t} + e^{2\tilde{t}}\delta_{ij}dx^I \otimes dx^J$,
where now $\tilde{t} \in [0,\infty)$.
}
\begin{equation} \label{FLRW}
\gt = \frac{1}{t^2} g
\end{equation}
with
\begin{equation} \label{conformal}
g = -dt\otimes dt + \delta_{IJ}dx^I \otimes dx^J.
\end{equation}
 It is important to note that, due to our conventions, the future is located in the direction of \textit{decreasing} $t$ and future timelike infinity is located at $t=0$. 
Consequently, we require that $dt(\vt)<0$
holds in order to guarantee that the four-velocity is future directed. 

From a cosmological perspective, the following family of solutions to the relativistic Euler equations are distinguished because they are spatially isotropic and homogeneous:
\begin{equation} \label{Hom-A}
(\rho_H,v_H^i) = (t^{3(1+\sspeed)}\rho_c,-\delta^i_0), \quad t\in (0,1],
\end{equation} 
where $\rho_c>0$ is any positive constant and $v^i_H$ is the conformal four velocity, see \eqref{vi-def} below. For sound speeds satisfying $0<\sigma<\frac{1}{3}$, the nonlinear stability to the future of this distinguished family of solutions was first established in 
the articles\footnote{In these articles, stability was established for relativistic, self-gravitating fluids with a positive cosmological constant. However, the techniques developed there also apply to the simpler setting considered in this article where gravitational effects are neglected.}
\cite{RodnianskiSpeck:2013,Speck:2012}. Stability for the end points $\sigma=0$ and $\sigma=\frac{1}{3}$ was later established in\footnote{Again, stability was established in these articles in the more difficult case of relativistic, self-gravitating fluids with a positive cosmological constant.} \cite{LubbeKroon:2013} and \cite{HadzicSpeck:2015},
respectively. We further note the articles \cite{Friedrich:2017,LiuOliynyk:2018b,LiuOliynyk:2018a,Oliynyk:CMP_2016} where other approaches for proving stability were developed, and the articles \cite{LeFlochWei:2021,LiuWei:2021} in which the authors analysed the stability of fluids with nonlinear equations of state. Additionally, we mention the articles \cite{Fajman_et_al:2023,FOW:2021,Ringstrom:2008,Ringstrom:2009,Speck:2013,Wei:2018} where related stability results are established.

On the other hand, for sound speed satisfying $\frac{1}{3}<\sspeed <1$, it was conjectured by Rendall in \cite{Rendall:2004} that the homogeneous and isotropic family of solutions \eqref{Hom-A} are unstable\footnote{While Rendall made this conjecture for relativistic, self-gravitating fluids with a positive cosmological constant, the same arguments apply to the simpler setting we are considering in this article.} to the future in the sense that inhomogeneous features will appear in the fluid density of nonlinear perturbations that will lead to blow-up of the fractional density gradient at future timelike infinity. This blow-up was first observed numerically in the article \cite{MarshallOliynyk:2022} to occur in $\Tbb^2$-symmetric solutions of the relativistic Euler equations on the same exponentially expanding FLRW spacetimes that we consider in this article. Later, the same behaviour was observed in solutions of the Gowdy symmetric Einstein-Euler equations in the article \cite{BMO:2023}, which numerically confirmed Rendall's fractional density gradient blow-up scenario.
It is worthwhile noting here that a related instability in spherically symmetric solutions of the Einstein-Euler equations has been observed for sound speed satisfying $\frac{1}{3}<\sspeed<1$ \cite{Harada:2001}. We further note that stability for the end point $\sigma=1$, which corresponds to a stiff fluid, was analysed in the article \cite{Fournodavlos:2022}. There, it was shown that the fluid four-velocity of every, sufficiently small irrotational perturbation of a FLRW fluid solution  becomes spacelike in finite time. This result rigorously proves that the FLRW solutions are nonlinearly  unstable to the future for $\sigma=1$.

An important point to make clear at this point is that not all solutions to the relativistic Euler equations are expected to be unstable for sound speed satisfying $\frac{1}{3}<\sigma < 1$. Indeed, it was shown in  \cite{MarshallOliynyk:2022,Oliynyk:2021} that there exists a family of spatially homogeneous, but not isotropic, solutions (so-called \textit{tilted fluids}) to the relativistic Euler equations on exponentially expanding FLRW spacetimes that are nonlinearly stable to the future for sound speed satisfying $\frac{1}{3}<\sigma < 1$.

Now, the main observations from \cite{MarshallOliynyk:2022} that were drawn from the numerical study of the $\Tbb^2$-symmetric relativistic Euler
can be summarised as follows:

\begin{enumerate}[(A)]
    \item For each $\sspeed\in(\frac{1}{3},1)$ and each choice of $\Tbb^2$-symmetric initial data that is close to the spatially homogeneous and isotropic data
     \begin{equation*}
(\rho_H,v_H^i)\bigl|_{t=1} = (\rho_c,-\delta^i_0),
    \end{equation*}
     the $\Tbb^2$-symmetric numerical solution of the relativistic Euler equations that is generated from this initial data displays ODE dominated behaviour at late times and is well-approximated by a solution of the \textit{asymptotic equation}, which is obtained by discarding all spatial derivatives from the relativistic Euler equations.
    \item For each $\sspeed \in(\frac{1}{3},1)$, there exists $\Tbb^2$-symmetric initial data, which is close to spatially homogeneous and isotropic data
    and for which the component of the spatial fluid velocity in the direction orthogonal to the $\Tbb^2$-symmetry crosses zero at a finite number of points, that generates $\Tbb^2$-symmetric numerical solutions of the relativistic Euler equations whose fractional density gradient develops sharp features and becomes unbounded at future timelike infinity. This realises numerically, in the restricted setting of relativistic fluids on expanding cosmological spacetime, the fractional density contrast blow-up conjectured by Rendall.
\end{enumerate}
In interpreting observation (A) above, it is important to emphasise that the appearance of ODE dominated behavior, if it exists, can only be be observed by using a suitable (non-unique) choice of the variables to represent the relativistic fluid. Each choice of variables not only determines the form of the evolution equation, but also the form of the associated asymptotic equation. Consequently, different choices of variables will lead to different evolution and asymptotic equations. The choice of variables, and hence evolution and asymptotic equations, employed in  \cite{MarshallOliynyk:2022}, see equations  (1.62)-(1.67)  from that article, is different from the one made here, see \eqref{Wtt-def}, \eqref{T2-Eul-E} and \eqref{T2-Eul-asymp.1}. While the choice of variables is different, the mathematical and physical content is the same. 

The main aim of this article is, for sounds speeds satisfying $\frac{1}{3}<\sspeed <1 $ and under additional technical restrictions, to prove that nonlinear $\Tbb^2$-symmetric perturbations of the spatially homogeneous and isotropic solutions \eqref{Hom-A} of the relativistic Euler equations display behaviour that agrees with observations (A) and (B) from above. The precise statements of the main results obtained in this article can be found in Theorem \ref{thm-exist}, Remark \ref{rem-exist} and Theorem \ref{thm-non-empty}. Theorem \ref{thm-exist} establishes that $\Tbb^2$-symmetric, nonlinear perturbations of the spatially homogeneous and isotropic solutions \eqref{Hom-A} exist globally to the future and determines their leading order asymptotics,  while Theorem \ref{thm-non-empty} establishes the existence of initial data for which the fractional density contrast blows-up at future timelike infinity thereby rigorously confirming observation (B) and a conjecture of Rendall from \cite{Rendall:2004}. 
Observation (A) is rigorously justified in Remark \ref{rem-exist}.(iv).

One intriguing aspect regarding Theorem \ref{thm-exist} is that we can only show that the $H^k(\Tbb)$ norm of the perturbed solutions remains bounded all the way to future timelike infinity  provided that sound speed satisfies $\frac{1}{3}<\sspeed <\frac{1+k}{3 k}$. Thus for fixed $\sigma \in (\frac{1}{3},1)$, there appears to be an upper bound $\kappa_{\sspeed}$ satisfying  
\begin{equation} \label{kappa-sspeed-def}
\kappa_{\sspeed} \geq \frac{1+k}{3 k}
\end{equation}
on the differentiability for which the solution in the $H^k(\Tbb)$ norm remains uniformly bounded to the future if $k < \kappa_{\sspeed}$ ($k \leq \kappa_{\sspeed}$), and  presumably, blows-up at future timelike infinity if $k\geq \kappa_{\sspeed}$ ($k > \kappa_{\sspeed}$). If this behaviour actually occurs, it would then imply a $H^k(\Tbb)$-instability for sufficiently high differentiability where the instability is due to \textit{norm inflation}. While we do not prove that norm inflation actually occurs in this article, it was observed in the numerical solutions from \cite{MarshallOliynyk:2022} and so we expect that it is a real phenomena. 

It is worthwhile noting here that the blow-up of the fractional density contrast at future timelike infinity is of particular physical interest because it is not consistent with the expected behaviour of fluids on cosmological spacetimes; see \cite[\S 5]{BMO:2023} for a discussion of the expected behaviour. Indeed, the fractional density gradient is expected to remain uniformly bounded and converge to a bounded function at future timelike infinity. This has been rigorously verified for perturbations of the spatially homogeneous and isotropic solutions \eqref{Hom-A} to the relativistic Euler equations, e.g. see \cite{Speck:2013}, for sounds speeds satisfying $0\leq \sspeed \leq \frac{1}{3}$ and also for the tilted family of spatially homogeneous (but not isotropic) solutions to the relativistic Euler equations studied in \cite{MarshallOliynyk:2022,Oliynyk:2021} for sounds speeds satisfying $\frac{1}{3}< \sspeed <1$.

For the remainder of the article, we will always assume that the sound speed satisfies
\begin{equation*}
\frac{1}{3} < \sspeed < 1.
\end{equation*}

\subsection{Future directions}
A number of interesting lines of investigation are suggested by the results of this article. Here, we mention two that we are actively pursuing. First, we would like to remove the $\Tbb^2$-symmetry assumption, and once that is accomplished, consider coupling to the Einstein equations with a positive cosmological constant. One goal of this line of investigation would be to rigorously verify the blow-up at future timelike infinity of the fractional density gradient that was observed in numerical solutions to the Einstein-Euler equations in \cite{BMO:2023}. Doing so would fully confirm Rendall's conjecture. 

The second direction we are currently pursuing is to extend the results of this article to allow for fractional Sobolev regularity; in this article, we only work with Sobolev spaces of integral regularity. This extension will be necessary to have any chance of determining the critical regularity $\kappa_{\sspeed}$, c.f.~\eqref{kappa-sspeed-def}, past which norm inflation can occur.

\subsection{Overview} This article is organised as follows. First, in Section \ref{sec:prelim}, we fix notation that will be employed throughout. Next, in Section \ref{relEulsec}, we introduce a symmetric hyperbolic formulation of the relativistic Euler equations, see \eqref{relEulB}, and then in Section \ref{sec:T2-symmetric}, we derive a $\Tbb^2$-symmetric version of these equations, see \eqref{T2-Eul-B.1}-\eqref{T2-Eul-B.2}. 
In Section \ref{sec:T2-symmetric}, we introduce a nonlinear change of variables, see \eqref{cov} and \eqref{phi-def}, and a $t$-weighted rescaling, see \eqref{Wtt-def}, in order to bring the $\Tbb^2$-symmetric relativistic Euler equations into a Fuchsian form, see \eqref{T2-Eul-E} and \eqref{T2-Eul-F}, that is suitable for analysing solutions all the way to future timelike infinity located at $t=0$. In the following section, Section \ref{sec:coeff-bounds}, bounds on the coefficient matrices $\Bt(t,\wtt)$, $\Att^0(t,\wtt)$ and $\Att^1(t,\wtt)$, see \eqref{AcBc-def}, that appear in the Fuchsian formulations \eqref{T2-Eul-E} and \eqref{T2-Eul-F} of the $\Tbb^2$-symmetric relativistic Euler equations are derived along with a commutator bound. These bounds are then used in Section \ref{sec:exist-asymp} to establish the future global existence and uniqueness of solutions to the $\Tbb^2$-symmetric relativistic Euler equations that correspond to nonlinear perturbations of the spatially homogeneous and isotropic family of solutions \eqref{Hom-A}. The precise statement of our future global existence result is given in Theorem \ref{thm-exist}. In addition to establishing the future global existence of solutions, we also, in Section \ref{sec:exist-asymp}, determine leading order asymptotics for these solutions near future timelike infinity; see Theorem \ref{thm-exist} for details. In Section \ref{sec:frac-den-grad-blow-up}, we establish the existence of open sets of initial data that lead to, among other things, the blow-up at future timelike infinity of the fractional density gradient. Finally, in Section \ref{sec:beyond}, we derive  beyond leading order asymptotic expansions and use them to determine in greater detail the behaviour of solutions near timelike infinity.

\section{Preliminaries\label{sec:prelim}}

\subsection{Coordinates and indexing conventions\label{indexing}}
We define the $3$-torus by $\Tbb^3=(\Tbb)^3$ where $\Tbb$ is obtained by quotienting $\Rbb$, i.e. $\Tbb=\Rbb/\sim$, by the equivalence relation: $x \sim y$ if and only if $y = x + n L$ where $n\in \Zbb$ and $L$, the period, is a fixed number in $\Rbb$. On the spacetime manifold \eqref{M-def},
we use $(x^i)=(x^0,x^I)$ to denote coordinates where the $(x^I)$ are periodic spatial coordinates
on $\Tbb^{3}$ and $x^0$ is a Cartesian time coordinate
on the interval $(0,1]$. Lower case Latin letters, e.g. $i,j,k$,
will label spacetime coordinate indices that run from $0$ to $3$ while upper case Latin letters, e.g. $I,J,K$, will label spatial coordinate indices that run from
$1$ to $3$. Partial derivatives with respect to the coordinates $(x^i)$ will be denoted by $\del{i} = \frac{\del{}\;}{\del{}x^i}$.
We will frequently set $t=x^0$, and use the notion $\del{t} = \del{0}$ for the partial derivative with respect to the coordinate $x^0$.

\subsection{Inner-products and matrices}
We denote the Euclidean inner-product on $\Rbb^n$ by 
\begin{equation*}
\ipe{\xi}{\zeta} = \xi^{\tr} \zeta, \quad \xi,\zeta \in \Rbb^n,
\end{equation*} 
and use $|\xi| = \sqrt{\ipe{\xi}{\xi}}$
to denote the Euclidean norm. Letting $\Mbb{n}$ denote the set of $n\times n$-matrices, we, for any $A,B\in \Mbb{n}$, define 
\begin{equation*}
    A\leq B \quad \Longleftrightarrow \quad \ipe{\xi}{A\xi} \leq \ipe{\xi}{B\xi}, \quad \forall \; \xi \in \Rbb^n.
\end{equation*}
Also, for $A\in \Mbb{n}$, we define
the operator norm in the usual fashion via
\begin{equation*}
   |A|_{\op} = \sup_{\xi\in \Rbb^n_\times} \frac{|A\xi|}{|\xi|}
\end{equation*}
where $\Rbb^n_\times = \Rbb^n\setminus\{0\}$.

\subsection{Constants and inequalities}
We use the standard notation $a \lesssim b$
for inequalities of the form
$a \leq Cb$
in situations where the precise value or dependence on other quantities of the constant $C$ is not required.
On the other hand, when the dependence of the constant on other inequalities needs to be specified, for
example if the constant depends on the norm $\norm{u}_{L^\infty}$, we use the notation
$C=C(\norm{u}_{L^\infty})$.
Constants of this type will always be non-negative, non-decreasing, continuous functions of their argument.

\subsection{Sobolev spaces}
The $W^{k,p}$, $k\in \Zbb_{\geq 0}$, norm of a map $C^\infty(\Tbb,\Rbb^n)$ is defined by
\begin{equation*}
\norm{u}_{W^{k,p}(\Tbb)} = \begin{cases} \begin{displaystyle}\biggl( \sum_{\ell=0}^{k} \int_{\Tbb} |\del{x}^\ell u|^p \, dx\biggl)^{\frac{1}{p}}  \end{displaystyle} & \text{if $1\leq p < \infty $} \\
 \begin{displaystyle} \max_{0\leq \ell \leq k}\sup_{x\in \Tbb}|\del{x}^\ell u(x)|  \end{displaystyle} & \text{if $p=\infty$}
\end{cases}.
\end{equation*}
The Sobolev space $W^{k,p}(\Tbb,\Rbb^n)$ is then defined to be the completion of $C^\infty(\Tbb,\Rbb^n)$ with respect to the norm
$\norm{\cdot}_{W^{k,p}}$. When $n=1$ or the dimension $n$ is clear from the context, we simplify the notation and write $W^{k,p}(\Tbb)$ instead, and we employ the standard notation $H^k(\Tbb,\Rbb^n)=W^{k,2}(\Tbb,\Rbb^n)$ throughout.
On $\Tbb$, we denote the $L^2$ inner-product by
\begin{equation*}
\ip{u}{v}= \int_{\Tbb} \ipe{u(x)}{v(x)}\, dx.
\end{equation*}


\section{A symmetric hyperbolic formulation of the Relativistic Euler equations\label{relEulsec}}
The starting point for the existence results established in this article is the symmetric hyperbolic formulation of the relativistic Euler equations from \cite[\S 2]{Oliynyk:2021}. These equations are formulated in terms of the spatial components $v_J$ of the \textit{conformal four velocity} 
\begin{equation} \label{vi-def}
v_i= \frac{1}{t} g_{ij}\vt^{j} 
\end{equation}
and the \textit{modified density} $\zetat$ that is defined via
\begin{equation}\label{zetat-def}
 \rho = t^{3(1+\sspeed)}\rho_c v_0^{-(1+\sspeed)}e^{(1+\sspeed)\zetat},
\end{equation}
where $\rho_c>0$ is an arbitrary positive constant and we have set
\begin{equation}\label{v0-def}
v_0 = \sqrt{|v|^2 +1} , \qquad |v|^2 = \delta_0^{IJ}v_I v_J.
\end{equation}
As shown in  \cite[\S 2]{Oliynyk:2021}, c.f.~equations (2.15)-(2.17) from \cite{Oliynyk:2021}, the relativistic Euler equations \eqref{relEulA}
can be expressed in terms of the variables $(v_J,\zetat)$ as
\begin{equation} \label{relEulB}
\Bt^k \del{k}\Vt = \frac{1}{t} \frac{1-3\sspeed}{v_0^2}\Pi \Vt
\end{equation}
where 
\begin{align*}
\Vt &= (\zetat, v_J)^{\tr}, 
\\
\Bt^0 &= \begin{pmatrix} \sspeed & 0 \\ 0 &\dsp M^{IJ}-\frac{\sspeed}{v_0^4}v^I v^J\end{pmatrix} ,
\\
\Bt^K &=  -\frac{1}{v_0}\begin{pmatrix} \sspeed v^K & \sspeed M^{KJ} \\ 
  \sspeed M^{JK}&  \dsp \biggl(M^{IJ}+\frac{\sspeed}{v_0^4}v^I v^J\biggr)v^K -\frac{\sspeed}{v_0^2}
  (v^I \delta_0^{JK}+v^J \delta_0^{IK}) \end{pmatrix}, 
  \\
\Pi &=\begin{pmatrix} 0 & 0 \\ 0 & \delta_0^{IJ} \end{pmatrix},
\\
M^{IJ} &= \delta_0^{IJ} - \frac{1}{v_0^2}v^I v^J  
\intertext{and}
v^J &= \delta_0^{JI}v_I. 
\end{align*}

\section{The $\Tbb^2$-symmetric relativistic Euler equations\label{sec:T2-symmetric}} 
Following \cite[\S 5]{Oliynyk:2021}, see also \cite[\S 1.3]{MarshallOliynyk:2022}, we can obtain a $\Tbb^2$-symmetric reduction of the relativistic Euler equations \eqref{relEulB} from the ansatz
\begin{align}
\tilde \zeta(t,x^1,x^2,x^3)&=z(t,x^1), \label{zttt-def} \\
v_{I}(t,x^1,x^2,x^3) &= u(t,x^1) \delta_I^1. \label{vttt-def}
\end{align}
As can be verified by a direct calculation, the pair $(\zetat,v_I)$ will solve \eqref{relEulB} provided $(z,u)$ satisfy\footnote{Here and in the following, we set $x=x^1$ and $\del{x}=\del{x^1}$.}
\begin{align}
\del{t}z-\frac{u }{\sqrt{1+u^2}}\del{x}z-\frac{1}{\left(1+u^2\right)^{3/2}}\del{x}u&=0, \label{T2-Eul-A.1} \\
\del{t} u-\frac{\sspeed \sqrt{1+u^2} }{1+(1-\sspeed)
   u^2}\del{x}z+\frac{(2 \sspeed-1-(1-\sspeed) u^2) u}{\sqrt{1+u^2}
   (1+(1-\sspeed) u^2)} \del{x}u&=\frac{1}{t} \frac{(1-3 \sspeed)(1+u^2) u}{1+(1-\sspeed) u^2}. \label{T2-Eul-A.2}
\end{align}

\subsection{A nonlinear change of velocity variable\label{sec:cov}}
The $\Tbb^2$-symmetric relativistic Euler equations \eqref{T2-Eul-A.1}-\eqref{T2-Eul-A.2}  are symmetisable, and consequently, the local-in-time existence of solutions, that is, on time intervals $(t_1,1]$ with $t_1$ sufficiently close to $1$, is guaranteed by standard existence theory for symmetric hyperbolic systems. However, the form of these equations is not useful for analysing nonlinear perturbations of the homogeneous solution \eqref{Hom-A} for sound speed satisfying $1/3<\sspeed < 1$ on the time interval $(0,1]$. This is because the singular source term
on the right hand side of \eqref{T2-Eul-A.2} is nonlinear in $u$ and has the wrong sign, i.e. $1-3\sspeed<0$. Due to these properties, the source term produces dangerous singular terms with unfavourable signs when using \eqref{T2-Eul-A.1}-\eqref{T2-Eul-A.2} to derive energy estimates for $(z,u)$ and its derivatives. 

To overcome this difficulty, we linearise the source term using a change of the velocity variable of the form 
\begin{equation}\label{cov}
w=\phi(u).
\end{equation}
Under this change of variable, a short calculation shows that the system \eqref{T2-Eul-A.1}-\eqref{T2-Eul-A.2} transforms as
 \begin{align}
\del{t}z-\frac{u }{\sqrt{1+u^2}}\del{x}z-\frac{1}{(1+u^2)^{\frac{3}{2}}}\Bigl(\frac{d\phi}{du}\Bigr)^{-1}\del{x}w&=0, \label{T2-Eul-B.1} \\
\del{t} w-\frac{\sspeed \sqrt{1+u^2} }{1+(1-\sspeed)
   u^2}\frac{d\phi}{du}\del{x}z+\frac{(2 \sspeed-1-(1-\sspeed) u^2) u}{\sqrt{1+u^2}
   (1+(1-\sspeed) u^2)} \del{x}w&=\frac{1}{t} \frac{d\phi}{du}\frac{(1-3 \sspeed)(1+u^2) u}{1+(1-\sspeed) u^2}. \label{T2-Eul-B.2}
\end{align} 
We linearise the source term of \eqref{T2-Eul-B.2} by demanding that $\phi(u)$ solve the ODE
\begin{equation*}
\frac{d\phi}{du}=\frac{1+(1-\sspeed) u^2}{(1+u^2) u} \phi. 
\end{equation*}
This equation is separable and we can directly integrate it to get the particular solution
\begin{equation}\label{phi-def}
\phi(u) = \frac{u}{(1+u^2)^{\frac{\sspeed}{2}}}. 
\end{equation}
Because we cannot invert \eqref{phi-def} analytically, we first establish that it is smoothly invertible in
the following lemma in order to ensure the change of variable 
\eqref{cov} is well-defined. 

\begin{lem}\label{phi-lem}
Suppose $\sspeed \in [0,1)$. Then the map
\begin{equation*}
\phi\: : \: \Rbb \longrightarrow \Rbb \: :\: u\longmapsto  \frac{u}{(1+u^2)^{\frac{\sspeed}{2}}}
\end{equation*}
is a monotonically increasing diffeomorphism. Moreover, for any $w_0<0<w_1$, the inverse map $\phi^{-1}(w)$ and
its derivative admit the asymptotic expansions:
\begin{align*}
\phi^{-1}(w)&=w^{\frac{1}{1-\sspeed}}\Bigl(1+ \Ord\Bigl(|w|^{\frac{2}{\sspeed-1}}\Bigr)\Bigr), \quad w>w_1, \\
\phi^{-1}(w)&=-|w|^{\frac{1}{1-\sspeed}}\Bigl(1+ \Ord\Bigl(|w|^{\frac{2}{\sspeed-1}}\Bigr)\Bigr), \quad w<w_0, \\
(\phi^{-1})'(w)&=w^{\frac{\sspeed}{1-\sspeed}}\biggl(\frac{1}{1-\sspeed}+ \Ord\Bigl(|w|^{\frac{2}{\sspeed-1}}\Bigr)\biggr), \quad w>w_1,
\intertext{and}
(\phi^{-1})'(w)&=|w|^{\frac{\sspeed}{1-\sspeed}}\Bigl(\frac{1}{1-\sspeed}+ \Ord\Bigl(|w|^{\frac{2}{\sspeed-1}}\Bigr)\Bigr), \quad w<w_0.
\end{align*}
\end{lem}
\begin{proof}
$\;$

\noindent\underline{Diffeomorphism:}
Since $\sspeed \in [0,1)$, we observe from \eqref{phi-def} that $\phi \in C^\infty(\Rbb,\Rbb)$, $\lim_{u\rightarrow \pm\infty}\phi(u) = \pm \infty$ and 
\begin{equation} \label{dphi}
\frac{d\phi}{du} = \frac{1+(1-\sspeed)u^2}{(1+u^2)^{\frac{2+\sspeed}{2}}}> 0, \quad u\in \Rbb.
\end{equation} 
From this, we can immediately conclude that $\phi$ is a monotonically increasing diffeomorphism from $\Rbb$ to $\Rbb$, and moreover that its inverse $\phi^{-1}$ is also monotonically increasing.

\medskip
\noindent\underline{Asymptotics:} Since $\phi(0)=0$ by \eqref{phi-def}, we deduce that
\begin{equation} \label{w-lbnd}
0<u_1=\phi^{-1}(w_1) < u=\phi^{-1}(w), \quad 0<w_1<w.
\end{equation}
From this, we obtain
\begin{equation*}
w=\phi(u) = \frac{u}{(1+u^2)^{\frac{\sspeed}{2}}} = \frac{u^{1-\sspeed}}{\dsp\biggl(1+\frac{1}{u^2}\biggr)^{\frac{\sspeed}{2}}},
\end{equation*}
which we can rearrange to get
\begin{equation} \label{u-asymp-A}
u = w^{\frac{1}{1-\sspeed}}\biggl(1+\frac{1}{u^2}\biggr)^{\frac{\sspeed}{2(1-\sspeed)}}.
\end{equation}
Expanding the term in the brackets gives
\begin{equation} \label{u-asymp-B}
u = w^{\frac{1}{1-\sspeed}}\biggl(1+ \frac{f(u^{-2})}{u^2}\biggr)
\end{equation}
where $f(\xi)$ is analytic on the region  $|\xi|< 2/u_1^2$ and satisfies $f(0)=\frac{\sspeed}{\sspeed(1-\sspeed)}$.
Recalling that $1-\sspeed>0$, the asymptotic expansion
\begin{equation*}
u = \phi^{-1}(w)=w^{\frac{1}{1-\sspeed}}\Bigl(1+ \Ord\Bigl(|w|^{\frac{2}{\sspeed-1}}\Bigr)\Bigr), \quad w>w_1,
\end{equation*}
then follows from plugging the representation \eqref{u-asymp-A} for $u$ into the right hand side of \eqref{u-asymp-B}.

Next, by \eqref{dphi}, we observe that
\begin{equation*}
\frac{du}{dw}= (\phi^{-1})'(w) = \biggl(\frac{d\phi}{du}\biggr)^{-1}=  \frac{(1+u^2)^{\frac{2+\sspeed}{2}}}{1+(1-\sspeed)u^2}
=   \frac{\dsp u^\sigma \biggl(1+\frac{1}{u^2}\biggr)^{\frac{2+\sspeed}{2}}}{\dsp (1-\sspeed)+\frac{1}{u^2}}.
\end{equation*}
Expanding the right hand side, we can express this as 
\begin{equation*}
\frac{du}{dw}=u^\sigma \biggl(\frac{1}{1-\sspeed} + \frac{h(u^{-2})}{u^2}\biggr)
\end{equation*}
where $h(\xi)$ is analytic on the region  $|\xi|< 2/u_1^2$ and satisfies $h(0)=-\frac{\sspeed(1+\sspeed)}{2(1-\sspeed)^2}$.
Plugging the representation \eqref{u-asymp-A} of $u$ into the right hand side of the this expression yields the asymptotic 
expansion
\begin{equation*}
\frac{du}{dw}=(\phi^{-1})'(w)=w^{\frac{\sspeed}{1-\sspeed}}\biggl(\frac{1}{1-\sspeed}+ \Ord\Bigl(|w|^{\frac{2}{\sspeed-1}}\Bigr)\biggr), \quad w>w_1>0.
\end{equation*}
This establishes the validity of the asymptotic expansions on the region $w>w_1>0$. Noting that $\phi(u)$ satisfies $\phi(-u)=-\phi(u)$,
it follows that $\phi^{-1}(-w)=-\phi^{-1}(w)$, and after differentiating, that $(\phi^{-1})'(-w)=(\phi^{-1})'(w)$. Using these relations, the asymptotic expansions on the region $w<w_0<0$ can be obtained directly from the asymptotic expansions on the region $w>w_1=|w_0|>0$.
\end{proof}

To proceed, we employ \eqref{phi-def} and \eqref{dphi} to express the $\Tbb^2$-symmetric relativistic Euler equations \eqref{T2-Eul-B.1}-\eqref{T2-Eul-B.2} in matrix form as
\begin{equation}\label{T2-Eul-C}
\del{t}W + B\del{x}W = \frac{1-3\sspeed}{t}\pi W
\end{equation}
where
\begin{align}
W&= (z,w)^{\tr}, \label{W-def}\\
B &=  \begin{pmatrix} \dsp
-\frac{u }{\sqrt{1+u^2}}&\dsp -\frac{1}{(1+u^2)^{\frac{1-\sspeed}{2}} (1+(1-\sspeed)u^2 )}\\[4mm]
\dsp -\frac{\sspeed}{(1+u^2)^{\frac{1+\sspeed}{2}}}&\dsp \frac{(2 \sspeed-1-(1-\sspeed) u^2) u}{\sqrt{1+u^2}
   (1+(1-\sspeed) u^2)}
\end{pmatrix}, \label{B-def}
\intertext{and}
\pi &= \begin{pmatrix}0 & 0 \\ 0 & 1 \end{pmatrix}. \label{pi-def}
\end{align} 
We further observe that
\begin{equation} \label{A0-def}
A^0 = \begin{pmatrix}
\dsp \frac{\sspeed (1+ (1-\sspeed)u^2)}{(1+u^2)^{\sspeed}} & 0 \\[2mm] 0 & 1
\end{pmatrix}
\end{equation}
is symmetriser for \eqref{T2-Eul-C}, and multiplying \eqref{T2-Eul-C} on the left with $A^0$ yields the symmetric hyperbolic system
\begin{equation}\label{T2-Eul-D}
A^0\del{t}W + A^1\del{x}W = \frac{1-3\sspeed}{t}\Pi W
\end{equation}
where
\begin{equation}\label{A1-def}
A^1 := A^0B =  \begin{pmatrix} \dsp
-\frac{\sspeed (1+ (1-\sspeed)u^2) u }{(1+u^2)^{\frac{1+2\sspeed}{2}}}&\dsp  -\frac{\sspeed}{(1+u^2)^{\frac{1+\sspeed}{2}}}\\[4 mm]
\dsp -\frac{\sspeed}{(1+u^2)^{\frac{1+\sspeed}{2}}}&\dsp \frac{(2 \sspeed-1-(1-\sspeed) u^2) u}{\sqrt{1+u^2}
   (1+(1-\sspeed) u^2)}
\end{pmatrix}.
\end{equation}

\begin{rem}\label{B-rem} 
Through the formula \eqref{B-def}, we have defined $B$ as a matrix valued map that depends on $u$. However, we will need to, at various points in the arguments below,
interpret $B$ as depending on  
$w$ via 
\begin{equation*}
B(w):= B\bigl|_{u=\phi^{-1}(w)}.
\end{equation*}
It will always be clear from context to which interpretation we are referring. The matrices $A^0$ and $A^1$ will be treated in a similar fashion as well as quantities derived from them and $B$. 
\end{rem}

\subsection{Fuchsian formulation\label{sec:rescale}}
While \eqref{T2-Eul-C} (or equivalently \eqref{T2-Eul-D}) is a Fuchsian formulation of the $\Tbb^2$-symmetric relativistic Euler equations, it not yet in a form that is suitable for analysing solutions all the way to future timelike infinity. To complete the transformation into a suitable form, we rescale $W$ by $t^{3\sspeed-1}$ to get
\begin{equation}\label{Wtt-def}
\Wtt = (\ztt,\wtt)^{\tr} := t^{3\sspeed-1} W = (t^{3\sspeed-1}z,t^{3\sspeed -1}w)^{\tr}.
\end{equation}
A short calculation then shows that \eqref{T2-Eul-C} and \eqref{T2-Eul-D} can be expressed in terms of $\Wtt$ in Fuchsian form as
\begin{equation}\label{T2-Eul-E}
\del{t}\Wtt + \Btt(t,\wtt)\del{x}\Wtt = \frac{3\sspeed-1}{t}\pip \Wtt,
\end{equation}
or equivalently
\begin{equation}\label{T2-Eul-F}
\Att^0(t,\wtt)\del{t}\Wtt + \Att^1(t,\wtt)\del{x}\Wtt = \frac{3\sspeed-1}{t}\Att^0(t,\wtt)\pip \Wtt,
\end{equation}
where
\begin{gather}
\Btt(t,\wtt) = B( t^{1-3\sspeed}\wtt), \quad \Att^0(t,\wtt)= A^0( t^{1-3\sspeed}\wtt), \quad \Att^i(t,\wtt)=A^i(t^{1-3\sspeed}\wtt), \label{AcBc-def}
\intertext{and}
\pip = \begin{pmatrix} 1 & 0 \\ 0 & 0 \end{pmatrix}. \label{pip-def}
\end{gather}
For later use, we observe from  \eqref{pi-def}, \eqref{A0-def}, \eqref{AcBc-def} and \eqref{pip-def} that the following identities hold:
\begin{equation}\label{pip-Ac0-com}
\pi^2 = \pi, \quad \pi \pip = 0, \quad  (\pip)^2 = \pip, \quad \pi^{\tr}=\pi, \quad (\pip)^{\tr}=\pip \AND [\pip,\Att^0] = [\pi,\Att^0] = 0.
\end{equation}

The method we will employ to establish the existence and uniqueness to solutions to \eqref{T2-Eul-E} is to adapted the Fuchsian global existence and uniqueness theory first introduced in \cite{Oliynyk:CMP_2016} and further developed in the articles \cite{BeyerOliynyk:2024,BOOS:2021,FOW:2021}. 
An adaption is necessary because the matrices $\Att^0(t,\wtt)$ and $\Att^0(t,\wtt)$ do not satisfy the assumptions needed to apply the global existence theory from \cite{BeyerOliynyk:2024,BOOS:2021,FOW:2021}. In Section \ref{sec:coeff-bounds} below, we derive a number of estimates that will allow us to adapt the Fuchsian global existence theory so that it can be applied to \eqref{T2-Eul-E}. The proof of global existence is carried out in Section \ref{sec:exist-asymp}.

\subsubsection{Matrix divergence: $\Div\! \Att(t,\wtt)$}
In deriving energy estimates, we will need to consider the \textit{matrix divergence} $\Div\! \Att(t,\wtt)$ defined by the right hand side of
\begin{align}
\del{t}(\Att^0(t,\wtt(t,x))) + \del{x}(\Att^1(t,\wtt(t,x)))=&t^{1-3\sspeed}\Bigl(\del{t}\wtt+
\frac{1-3\sspeed}{t}\wtt\Bigr) \frac{dA^0}{dw}(t^{1-3\sspeed}\wtt) \notag \\
&\qquad + t^{1-3\sspeed}\del{x}\wtt \frac{dA^1}{dw}(t^{1-3\sspeed}\wtt), \label{Div-Ac-0}
\end{align}
where $\wtt(t,x)$ is obtained from a solution of \eqref{T2-Eul-E}. Noting by \eqref{cov}, \eqref{phi-def}, \eqref{dphi} and \eqref{Wtt-def} 
that  
\begin{align*}
\Bigl( t^{1-3\sspeed}\del{t}\wtt+
\frac{1-3\sspeed}{t}t^{1-3\sspeed}\wtt\Biggr)\frac{dA^0}{dw}(t^{1-3\sspeed}\wtt) =\del{t}w \frac{dA^0}{dw}(w)=\del{t} w  \Bigl(\frac{d\phi}{du}\Bigr)^{-1}
\frac{dA^0}{du},
\end{align*}
we observe from  \eqref{cov}, \eqref{phi-def}, \eqref{T2-Eul-C}, \eqref{A0-def} and \eqref{pip-def} that 
 \begin{align}\label{Div-Ac-1}
\Bigl( t^{1-3\sspeed}\del{t}\wtt+
\frac{1-3\sspeed}{t}t^{1-3\sspeed}\wtt\Bigr)\frac{dA^0}{dw}(t^{1-3\sspeed}\wtt) =\Bigl(a_1\del{x}z+a_2\del{x}w + 
\frac{1}{t}\bigl(-2(3\sspeed-1)+a_3\bigr)\Bigr)A^0\pip
\end{align}
where
\begin{align}
a_1 &=\frac{2 \sspeed u ((1-\sspeed)^2 u^2-2 \sspeed +1)}{\sqrt{1+u^2} (1+(1-\sspeed ) u^2)^2}, \label{a1-def}\\
a_2 &=-\frac{2 u^2 (1+u^2)^{\frac{\sspeed -1}{2}}((\sspeed-1 ) u^2+2 \sspeed -1)((1-\sspeed )^2 u^2-2\sspeed +1)}{(1+(1-\sspeed ) u^2)^3}, \label{a2-def}
\intertext{and}
a_3 & =\frac{2 (3 \sspeed-1) \left(1+u^2\right)}{\left(1+(1-\sspeed) u^2\right)^2} .\label{a3-def}
\end{align}
\begin{rem} It is worth pointing out that the non-negativity of function $a_3$ is essential for establishing energy estimates that yield uniform bounds on the time interval $(0,1]$. Specifically, the non-negativity of $a_3$ allows us to obtain the estimate \eqref{energy-est-C} below. Had $a_3$ been negative anywhere, then our whole approach to establishing existence of solutions on the interval $(0,1]$ would have failed. 
\end{rem}
Substituting \eqref{Div-Ac-1} into \eqref{Div-Ac-0} and labelling the result $\Div\! \Att(t,\wtt)$ yields
\begin{equation}\label{DivAc-def}
\Div\! \Att(t,\wtt) =t^{1-3\sspeed}\bigl(a_1\del{x}\ztt+a_2\del{x}\wtt\bigr)\Att^0\pip + 
\frac{1}{t}\bigl(-2(3\sspeed-1)+a_3\bigr)\Att^0\pip + t^{1-3\sspeed}\del{x}\wtt \frac{dA^1}{dw}(t^{1-3\sspeed}\wtt).
\end{equation}
Defined this way, we then have that
\begin{equation} \label{DivAc-on-shell}
\Div\!\Att(t,\wtt(t,x)) = \del{t}(\Att^0(t,\wtt(t,x))) + \del{x}(\Att^1(t,\wtt(t,x)))
\end{equation}
for  functions $\wtt(t,x)$ that are obtained from a solution of \eqref{T2-Eul-E}. For use below, we will need to express $\frac{dA^1}{dw}$, which appears in \eqref{DivAc-def}, as 
\begin{equation}\label{DwAc1-rep}
\frac{dA^1}{dw}= (A^0)^{\frac{1}{2}}Q (A^0)^{\frac{1}{2}}
\end{equation}
where
\begin{equation}  \label{sqrtA0}
 (A^0)^{\frac{1}{2}}  = \begin{pmatrix}
\dsp \frac{\sqrt{\sspeed} \sqrt{1+ (1-\sspeed)u^2}}{(1+u^2)^{\frac{\sspeed}{2}}} & 0 \\[2mm] 0 & 1
\end{pmatrix}
\end{equation}
is the square root of the matrix $A^0$ and
\begin{equation*} 
Q= (A^0)^{-\frac{1}{2}} \frac{dA^1}{dw}  (A^0)^{-\frac{1}{2}}= \biggl(\frac{d\phi}{du}\biggr)^{-1}(A^0)^{-\frac{1}{2}} \frac{dA^1}{du}  (A^0)^{-\frac{1}{2}}.
\end{equation*}
Using \eqref{dphi}, \eqref{A1-def} and \eqref{sqrtA0}, we find after a straightforward calculation that
\begin{equation}\label{DwAc1-rep-A}
Q=
\begin{pmatrix}
\dsp  \frac{ -2 (1-\sspeed)^2 u^4+(5 \sspeed-3) u^2-1}{(1+u^2)^{\frac{1-\sspeed}{2}}(1+ (1-\sspeed)u^2)^2} &\dsp
 \frac{\sqrt{\sspeed} (\sspeed+1) u
  }{ (1+u^2)^{\frac{1-\sspeed}{2}}(1+(1-\sspeed) u^2)^{\frac{3}{2}}} \\[5mm]
\dsp \frac{\sqrt{\sspeed} (\sspeed+1) u }{(1+u^2)^{\frac{1-\sspeed}{2}}(1+(1-\sspeed) u^2)^{\frac{3}{2}}} & \dsp
\frac{
   \sspeed^2 (3 u^2+2) u^2-2 \sspeed (u^4-1)-(1+u^2)^2}{(1+u^2)^{\frac{1-\sspeed}{2}}(1+(1-\sspeed) u^2)^3}
\end{pmatrix}.
\end{equation}

\section{Coefficient bounds and a commutator estimate\label{sec:coeff-bounds}}
Before we consider the existence of future global solutions to the system \eqref{T2-Eul-E} or equivalently \eqref{T2-Eul-F}, we first establish a number of bounds on the coefficients of these systems that will be needed in the existence proof. The required bounds are stated in Lemmas \ref{lem-a1a2a3Q}, \ref{lem-A0B} and \ref{lem-Ac0-lbnd}, and Corollary \ref{cor-Ac0Bc}. Additionally, we establish a variation on a standard commutator estimate in Lemma \ref{comm-lem} that will also be needed in our existence proof.
\begin{lem}\label{lem-a1a2a3Q}
Suppose $\sspeed \in [1/3,1)$. Then
\begin{equation*}
 0\leq a_3(w) \leq \frac{2(3\sspeed-1)}{(1-\sspeed)^2}
\end{equation*}
for all $w\in \Rbb$ and 
\begin{equation*}
\sup_{w\in \Rbb}\bigl(|a_1(w)|+|a_2(w)|+|Q(w)|_{\op}\bigr) \leq C
\end{equation*}
for some constant $C>0$.
\end{lem}
\begin{proof}
From \eqref{a1-def}, \eqref{a2-def} and \eqref{DwAc1-rep-A}, we observe that, for some $C>0$, the functions $a_1$ and $a_2$, and  the matrix $Q$ can be bounded by 
\begin{equation*}
|a_1|\leq \frac{C}{1+u^2} \AND |a_2|+|Q|_{\op} \leq \frac{1}{(1+u^2)^{\frac{1-\sspeed}{2}}}
\end{equation*}
for all $u\in \Rbb$. Since $3\sspeed-1\geq 0$, we further observe from  \eqref{a3-def} that $a_3$ can be bounded above and below
by
\begin{equation*}
0\leq a_3 \leq \frac{2(3\sspeed-1)}{(1-\sspeed)^2}
\end{equation*}
for all $u\in \Rbb$.
The stated bounds are then a consequence of the above inequalities and the change of variables defined by \eqref{cov} and \eqref{phi-def}; see also Lemma \ref{phi-lem}.
\end{proof}

\begin{lem}\label{lem-A0B}
Suppose $\sspeed \in [0,1)$. Then there exist constants $C(\ell)$, $\ell\in \Zbb_{\geq 0}$, such that
\begin{equation*}
\sup_{w\in \Rbb}\biggr|\frac{d^\ell B}{dw^\ell}(w)\biggr|_{\op} \leq C(\ell), \quad \ell\geq 0,
\end{equation*}
and
\begin{equation*}
\sup_{w\in\Rbb} \biggl|A^0(w) \frac{d^\ell B}{dw^\ell}(w)\biggr|_{\op} \leq C(\ell), \quad \ell>0.
\end{equation*}
\end{lem}
\begin{proof}
Suppose  $\sspeed \in [0,1)$ and $\ell \in \Zbb_{\geq 0}$. Then denoting the components of the $\ell$-th derivative of the matrix $B$ with respect to $u$ by
\begin{equation*}
\frac{d^\ell B}{du^\ell} = \begin{pmatrix} B^\ell_{11} & B^\ell_{12} \\[2mm] B^{\ell}_{21} & B^{\ell}_{22} \end{pmatrix},
\end{equation*}
it is not difficult to verify from \eqref{B-def} that there exist constants $C(\ell)>0$ such that we can bound the components  $B^\ell_{\alpha\beta}$ by
\begin{equation} \label{dBdu-bnds-A}
|B_{11}^0|+|B_{12}^0|+|B_{21}^0|+|B_{22}^0| \leq C(0)
\end{equation}
and
\begin{equation} \label{dBdu-bnds}
|B_{11}^\ell|+|B_{22}^\ell| \leq C(\ell)(1+u^2)^{-\frac{2+\ell}{2}}, \quad |B_{12}^\ell|\leq C(\ell)(1+u^2)^{-\frac{3-\sspeed + \ell}{2}}, \quad |B_{21}^\ell|\leq C(\ell)(1+u^2)^{-\frac{1+\sspeed+\ell}{2}},
\end{equation}
for all $u\in \Rbb$ and $\ell>0$. 
Then letting
\begin{equation*}
A^0\frac{d^\ell B}{du^\ell}=  \begin{pmatrix} A^\ell_{11} & A^\ell_{12} \\[2mm] A^{\ell}_{21} & A^{\ell}_{22} \end{pmatrix},
\end{equation*}
it follows immediately from  \eqref{A0-def} and the bounds \eqref{dBdu-bnds} that we can, increasing $C(\ell)$ if necessary,  bound the components of
$A^0\frac{d^\ell B}{du^\ell}$ by
 \begin{equation}  \label{A0dBdu-bnds}
|A_{12}^\ell|+|A_{21}^\ell| \leq C(\ell)(1+u^2)^{-\frac{1+\sspeed+\ell}{2}}, \quad |A_{11}^\ell|\leq C(\ell)(1+u^2)^{-\frac{2\sspeed + \ell}{2}}, \quad |A_{22}^\ell|\leq C(\ell)(1+u^2)^{-\frac{2+\ell}{2}}, 
\end{equation}
for all $u\in \Rbb$ and $\ell>0$.
Next, by the chain rule, we have that
\begin{equation}\label{dw-2-du}
\frac{d\;}{dw} = \biggl(\frac{d\phi}{du}\biggr)^{-1}\frac{d\;}{du}
\end{equation} 
where we note from \eqref{dphi} that, for each $m\in \Zbb_{\geq 0}$, the $m$-th derivative with respect to $u$ of $\bigl(\frac{d\phi}{du})^{-1}$ is bounded by
\begin{equation} \label{dphidu-bnds}
\biggl|\frac{d^m}{du^m}\biggl(\frac{d\phi}{du}\biggr)^{-1}\biggr| \leq C(m) (1+u^2)^{\frac{\sspeed-m}{2}}
\end{equation}
for all $u\in \Rbb$. The stated bounds for $\frac{d^\ell B}{dw^\ell}$ and $A^0\frac{d^\ell B}{dw^\ell}$  now follow easily
from the inequalities \eqref{dBdu-bnds-A}, \eqref{dBdu-bnds}, \eqref{A0dBdu-bnds} and \eqref{dphidu-bnds}, the relation \eqref{dw-2-du}, and Lemma \ref{phi-lem}.
\end{proof}

\begin{cor} \label{cor-Ac0Bc}
Suppose $\sspeed \in [0,1)$. Then there exist constants $C(\ell)$, $\ell\in \Zbb_{\geq 0}$, such that
\begin{equation*}
\sup_{(t,\wtt)\in (0,1]\times\Rbb} t^{\ell(3\sspeed-1)}\bigl|\del{\wtt}^\ell \Btt(t,\wtt)\bigr|_{\op} \leq C(\ell) , \quad \ell \geq 0,
\end{equation*}
and
\begin{equation*}
\sup_{(t,\wtt)\in (0,1]\times\Rbb} t^{\ell(3\sspeed-1)} \bigl|\Att^0(t,\wtt) \del{\wtt}^\ell \Btt(t,\wtt)\bigr|_{\op} \leq C(\ell), \quad  \ell > 0.
\end{equation*}
\end{cor}
\begin{proof}
By \eqref{Wtt-def} and \eqref{AcBc-def}, we have
\begin{equation*}
\del{\wtt}^\ell \Btt(t,\wtt) = t^{\ell(1-3\sspeed)}\frac{d^\ell B}{dw^\ell}(t^{(1-3\sspeed)}\wtt) \AND
\Att^0(t,\wtt)\del{\wtt}^\ell \Btt(t,\wtt) = t^{\ell(1-3\sspeed)}A^0(t^{(1-3\sspeed)}\wtt)\frac{d^\ell B}{dw^\ell}(t^{(1-3\sspeed)}\wtt).
\end{equation*}
The proof then follows from an application of Lemma \ref{lem-A0B}
\end{proof}

\begin{lem}\label{lem-Ac0-lbnd}
Suppose $\sspeed \in (0,1)$ and let
\begin{equation*}
\gamma = \begin{cases} \sspeed & \text{if  $0<\sspeed \leq 1/2$} \\
\sspeed^{2(1-\sspeed)}(1-\sspeed)^{2\sspeed -1} & \text{if $1/2 < \sspeed <1$} \end{cases}.
\end{equation*}
Then there exists a constant $C\geq 1$ such that
\begin{equation*}
\gamma \id \leq \Att^0(t,\wtt)
\leq C\bigl(1+t^{2(1-3\sspeed)}|\wtt|^2\bigr)\id
\end{equation*}
for all $(t,\wtt) \in (0,1]\times \Rbb$.
\end{lem}
\begin{proof}
Setting
\begin{equation*}
\alpha(u) = \frac{\sspeed (1+ (1-\sspeed)u^2)}{(1+u^2)^{\sspeed}},
\end{equation*}
we observe, since $\sspeed \in (0,1)$, that $\lim_{u  \rightarrow \pm\infty} \alpha(u)= \infty$. Moreover, differentiating $\alpha$, we find that
\begin{equation*}
\alpha'(u) = 0 \quad \Longleftrightarrow \quad \begin{cases} u=0 & \text{if  $0<\sspeed \leq 1/2$} \\
u=0, \, \pm \frac{\sqrt{2\sspeed-1}}{1-\sspeed} & \text{if $1/2 < \sspeed < 1$} \end{cases}.
\end{equation*}
From this, we conclude that the absolute minimum of $\alpha(u)$ on $\Rbb$ is given by the constant $\gamma$ defined in the statement of the lemma, and in particular that 
$\alpha(u)\geq \gamma >0$ for all $u\in \Rbb$.
Since $0< \gamma <1$, the inequality
\begin{equation*}
\gamma \id \leq \Att^0(t,\wtt), \quad (t,\wtt)\in (0,1]\times \Rbb,
\end{equation*}
is then a direct consequence of the change of variables defined by  \eqref{cov} and \eqref{phi-def}, and the definitions \eqref{A0-def},
\eqref{Wtt-def} and \eqref{AcBc-def}.

Next, expressing $\alpha$ as a function of $w$
via 
\begin{equation*}
\alpha(w) = \frac{\sspeed (1+ (1-\sspeed)
(\phi^{-1}(w))^2)}{(1+(\phi^{-1}(w))^2)^{\sspeed}},
\end{equation*}
it follows from Lemma \ref{phi-lem} 
that, for any $w_0>0$, $\alpha$ satisfies 
\begin{equation*}
\alpha(w) = |w|^2\Bigl(\sspeed(1-\sspeed) + \Ord\bigl(|w|^{\frac{2}{\sspeed-1}}\bigr)\Bigr), \quad |w|>w_0,
\end{equation*}
and consequently, there exists a constant $C\geq 1$
such that $\alpha$ can be bounded above by
$\alpha(w) \leq C(1+|w|^2)$ for all $w\in \Rbb$. 
From the definitions \eqref{A0-def},
\eqref{Wtt-def} and \eqref{AcBc-def}, it
then follows that
\begin{equation*}
\Att^0(t,\wtt) \leq C\bigl(1+t^{2(1-3\sspeed)}|\wtt|^2\bigr)\id, \quad
(t,\wtt) \in (0,1]\times \Rbb.
\end{equation*}
\end{proof}

In the next lemma, we establish a variation of a standard commutator estimate, e.g. see Theorem \ref{Product}.(i) below. The reason we need to consider this variation is because the matrix $\Att^0(t,\wtt)$
that appears in the $\Att^0(t,\wtt)[\del{x}^\ell,\Btt(t,\wtt)]V$
is not bounded and so it cannot be estimated separately from the commutator $[\del{x}^\ell,\Btt(t,\wtt)]V$. It is only when the product  $\Att^0(t,\wtt)[\del{x}^\ell,\Btt(t,\wtt)]V$ is estimated together that we can establish an effective bound.

\begin{lem}\label{comm-lem}
Suppose $\frac{1}{3}<\sspeed < 1$ and $\ell \in \Zbb_{>0}$. Then
\begin{equation*}
\bigl\|\Att^0(t,\wtt)[\del{x}^\ell,\Btt(t,\wtt)]V\bigr\|_{L^2} \lesssim  t^{\ell(1-3\sspeed)}\bigl( 1+ \norm{\wtt}_{W^{1,\infty}}^{\ell-1}\bigr)\bigl(\norm{\wtt}_{W^{1,\infty}}\norm{V}_{H^{\ell-1}}+\norm{\wtt}_{H^\ell}\norm{V}_{L^\infty}\bigr)
\end{equation*} 
for all $t\in (0,1]$, $\wtt\in W^{1,\infty}(\Tbb)\cap H^{\ell}(\Tbb)$ and $V\in L^\infty(\Tbb,\Rbb^2)\cap H^{\ell-1}(\Tbb,\Rbb^2)$.
\end{lem}
\begin{proof}
Suppose that $\ell \in \Zbb_{>0}$, $t\in (0,1]$,
$\wtt\in W^{1,\infty}(\Tbb)\cap H^{\ell}(\Tbb)$ and $V\in L^\infty(\Tbb,\Rbb^2)\cap H^{\ell-1}(\Tbb,\Rbb^2)$. Then
expanding out the commutator $\Att^0(t,\wtt)[\del{x}^\ell,\Btt(t,\wtt)]V$ yields
\begin{align}
\Att^0(t,\wtt)[\del{x}^\ell,\Btt(t,\wtt)]V &= \Att(t,\wtt)\Bigl(\del{x}^\ell\bigl(\Btt(t,\wtt)V\bigr)-\Btt(t,\wtt)\del{x}^\ell V\Bigr) \notag \\
&= \Att(t,\wtt)\Biggl(\sum_{m=0}^\ell \binom{\ell}{m} \del{x}^{\ell-m}(\Btt(t,\wtt))\del{x}^{m}V-\Btt(t,\wtt)\del{x}^\ell V\Biggr)
\notag \\
&= \sum_{m=0}^{\ell-1} \binom{\ell}{m} \Sc_m \label{comm-sum}
\end{align}
where
\begin{equation*}
\Sc_m = \Att(t,\wtt)\del{x}^{\ell-m}(\Btt(t,\wtt))\del{x}^{m}V.
\end{equation*}
To estimate the terms in the sum \eqref{comm-sum}, we consider two cases: 

\bigskip

\noindent\underline{Case 1: $m=\ell-1$}:
Setting $m=\ell-1$, we note that
\begin{equation*}
\Sc_{\ell-1} =  \Att(t,\wtt)\del{\wtt}\Btt(t,\wtt)\del{x}\wtt \del{x}^{\ell-1}V.
\end{equation*} 
With the help of H\"{o}lder's inequality, we can estimate the $L^2$-norm of 
$\Sc_{\ell-1}$ by
\begin{equation}\label{Sc-est-1}
\norm{\Sc_{\ell-1}}_{L^2}\leq \norm{\Att(t,\wtt)\del{\wtt}\Btt(t,\wtt)}_{L^\infty}\norm{\del{x}\wtt}_{L^\infty}\norm{\del{x}^{\ell-1}V}_{L^2} \lesssim  t^{1-3\sspeed}\norm{\wtt}_{W^{1,\infty}}\norm{V}_{H^{\ell-1}}
\end{equation}
where the second inequality is a consequence of Corollary \ref{cor-Ac0Bc}.

\bigskip

\noindent\underline{Case 2: $0\leq m<\ell-1$}:
Noting, in this case, that $\big(\frac{2(\ell-1)}{\ell-m-1}\bigr)^{-1}+ \bigl(\frac{2(\ell-1)}{m}\bigr)^{-1} = \frac{1}{2}$
and $1<\frac{2(\ell-1)}{\ell-m-1}, \frac{2(\ell-1)}{m}\leq \infty$, we find from an application of  H\"{o}lder's inequality that
the $L^2$-norm of $\Sc_m$ is bounded by
\begin{equation}\label{Sc-est-2}
\norm{\Sc_m}_{L^2} \leq \norm{\Att(t,\wtt)\del{x}^{\ell-m}(\Btt(t,\wtt))}_{L^{\frac{2(\ell-1)}{\ell-m-1}}}\norm{\del{x}^{m}V}_{L^{\frac{2(\ell-1)}{m}}}.
\end{equation}
We then proceed by estimating the right hand side of the above inequality using the following Gagliardo-Nirenberg interpolation inequality (Theorem \ref{GNII}):
\begin{equation} \label{GN-inq}
\norm{\del{x}^k f}_{L^{\frac{2 s}{k}}} \lesssim \norm{f}_{L^\infty}^{1-\frac{k}{s}}\norm{f}_{H^s}^{\frac{k}{s}},\quad 0\leq k < s.
\end{equation} 
With the help of this inequality, we can immediately estimate the second term on the right hand side of  \eqref{Sc-est-2} by
\begin{equation}\label{Sc-est-3}
\norm{\del{x}^{m}V}_{L^{\frac{2(\ell-1)}{m}}} \lesssim \norm{V}_{L^\infty}^{1-\frac{m}{\ell-1}} \norm{V}_{H^{\ell-1}}^{\frac{m}{\ell-1}}.
\end{equation}
To estimate the first term,
we adapt the proof of the Moser estimate from Proposition 3.9 of  \cite[Ch.~13]{TaylorIII:1996} by expanding
the derivative $\del{x}^{\ell-m}(\Btt(t,\wtt))$ using Fa\`{a} di Bruno's formula to get
\begin{equation*}
\del{x}^{\ell-m}(\Btt(t,\wtt))=C_{1\ldots 1}\del{\wtt}^{\ell-m}\Btt(t,\wtt)(\del{x}\wtt)^{\ell-m}
+\sum C_{k_1\ldots k_j}\del{\wtt}^j\Bc(t,\wtt)\del{x}^{k_1}\wtt\cdots \del{x}^{k_j}\wtt
\end{equation*}
where the sum is over the set of integers $k_i$, $1\leq i\leq j$, satisfying $k_1+\cdots+k_j=\ell-m$, $k_i\geq 1$, and $k_{i'}\geq 2$ for some $i'\in \{1,\ldots,j\}$.
Noting that
\begin{equation*}
\frac{1}{\frac{2(\ell-1)}{k_1}}+\ldots+  \frac{1}{\frac{2(\ell-1)}{k_{i'-1}}} +\frac{1}{\frac{2(\ell-1)}{k_{i'}-1}}+ \frac{1}{\frac{2(\ell-1)}{k_{i'+1}}}+\ldots +\frac{1}{\frac{2(\ell-1)}{k_j}}= \frac{1}{\frac{2(\ell-1)}{\ell-m-1}},
\end{equation*}
we use the triangle and generalised H\"{o}lder inequalities to bound $\del{x}^{\ell-m}(\Btt(t,\wtt))$ by
\begin{align}
 \norm{\Att(t,\wtt)\del{x}^{\ell-m}(\Btt(t,\wtt))}_{L^{\frac{2(\ell-1)}{\ell-m-1}}}
&\lesssim \bigl\|\Att(t,\wtt)\del{\wtt}^{\ell-m}\Btt(t,\wtt)(\del{x}\wtt)^{\ell-m}\bigr\|_{L^{\frac{2(\ell-1)}{\ell-m-1}}} \notag 
\\
 +&\sum \bigl\|\Att(t,\wtt)\del{\wtt}^j\Bc(t,\wtt)\del{x}^{k_1}\wtt\cdots \del{x}^{k_j}\wtt\bigr\|_{L^{\frac{2(\ell-1)}{\ell-m-1}}}\notag\\
\lesssim&\bigl\|\Att(t,\wtt)\del{\wtt}^{\ell-m}\Btt(t,\wtt)\bigr\|_{L^\infty}\norm{\del{x}\wtt}_{L^\infty} \bigl\|(\del{x}\wtt)^{\ell-m-1}\bigr\|_{L^{\frac{2(\ell-1)}{\ell-m-1}}} \notag 
\\
+&\sum\bigl\|\Att(t,\wtt)\del{\wtt}^j\Bc(t,\wtt)\bigr\|_{L^\infty} \bigl\|\del{x}^{k_1}\wtt\cdots \del{x}^{k_j}\wtt\bigr\|_{L^{\frac{2(\ell-1)}{\ell-m-1}}} \notag \\
\lesssim&\bigl\|\Att(t,\wtt)\del{\wtt}^{\ell-m}\Btt(t,\wtt)\bigr\|_{L^\infty}\norm{\del{x}\wtt}_{L^\infty} 
\norm{\del{x}\wtt}^{\ell-m-1}_{L^{2(\ell-1)}} \notag 
\\
+\sum\bigl\|\Att(t,\wtt)&\del{\wtt}^j\Bc(t,\wtt)\bigr\|_{L^\infty} \norm{\del{x}^{k_1}\wtt}_{L^{\frac{2(\ell-1)}{k_1}}}\cdots 
\norm{\del{x}^{k_{i'}}\wtt}_{L^{\frac{2(\ell-1)}{k_{i'}-1}}} \cdots \norm{\del{x}^{k_j}\wtt}_{L^{\frac{2(\ell-1)}{k_j}}}. \notag
\end{align}
By \eqref{GN-inq}, we then have
\begin{align*}
 \norm{\Att(t,\wtt)\del{x}^{\ell-m}(\Btt(t,\wtt))}_{L^{\frac{2(\ell-1)}{\ell-m-1}}} \lesssim&\bigl\|\Att(t,\wtt)\del{\wtt}^{\ell-m}\Btt(t,\wtt)\bigr\|_{L^\infty}\norm{\del{x}\wtt}_{L^\infty} \Bigl(\norm{\wtt}_{L^\infty}^{1-\frac{1}{\ell-1}}\norm{\wtt}_{H^{\ell-1}}^{\frac{1}{\ell-1}}\Bigr)^{\ell-m-1}
 \notag 
\\
+\sum\bigl\|\Att(t,\wtt)\del{\wtt}^j\Bc(t,\wtt)\bigr\|_{L^\infty} &\norm{\wtt}_{L^{\infty}}^{1-\frac{k_1}{\ell-1}}\norm{\wtt}_{H^{\ell-1}}^{\frac{k_1}{\ell-1}}\cdots 
\norm{\del{x}\wtt}_{L^{\infty}}^{1-\frac{k_{i'}-1}{\ell-1}}\norm{\del{x}\wtt}_{H^{\ell-1}}^{\frac{k_{i'}-1}{\ell-1}} \cdots \norm{\wtt}_{L^{\infty}}^{1-\frac{k_j}{\ell-1}}\norm{\wtt}_{H^{\ell-1}}^{\frac{k_j}{\ell-1}} \notag\\
\lesssim& \sum_{j=1}^{\ell-m}\bigl\|\Att(t,\wtt)\del{\wtt}^j\Bc(t,\wtt)\bigr\|_{L^\infty}
\norm{\wtt}_{W^{1,\infty}}^{j-\frac{\ell-m-1}{\ell-1}}\norm{\wtt}_{H^{\ell}}^{\frac{\ell-m-1}{\ell-1}}.
\end{align*}
Since $1-3\sspeed > 0$ and $t\in (0,1]$ by assumption, we conclude from the above inequality
and Corollary \ref{cor-Ac0Bc} that
\begin{equation}\label{Sc-est-4}
 \norm{\Att(t,\wtt)\del{x}^{\ell-m}(\Btt(t,\wtt))}_{L^{\frac{2(\ell-1)}{\ell-m-1}}} \lesssim t^{(\ell-m)(1-3\sspeed)}\bigl( 1+ \norm{\wtt}_{W^{1,\infty}}^{\ell-m-1}\bigr)\norm{\wtt}_{W^{1,\infty}}^{\frac{m}{\ell-1}}\norm{\wtt}_{H^{\ell}}^{\frac{\ell-m-1}{\ell-1}}, \quad 0\leq m <\ell-1.
\end{equation}
Together \eqref{Sc-est-2}, \eqref{Sc-est-3} and \eqref{Sc-est-4}
show that
\begin{equation}\label{Sc-est-5}
\norm{\Sc_m}\lesssim  t^{(\ell-m)(1-3\sspeed)}\bigl( 1+ \norm{\wtt}_{W^{1,\infty}}^{\ell-m-1}\bigr)
\norm{\wtt}_{W^{1,\infty}}^{\frac{m}{\ell-1}}\norm{\wtt}_{H^{\ell}}^{\frac{\ell-m-1}{\ell-1}}\norm{V}_{L^\infty}^{\frac{\ell-m-1}{\ell-1}} \norm{V}_{H^{\ell-1}}^{\frac{m}{\ell-1}}
\end{equation}
for $0\leq m <\ell-1$. 

\bigskip

We now complete the proof by applying the $L^2$-norm to \eqref{comm-sum} and using the triangle inequality in conjunction with the estimates \eqref{Sc-est-1}, \eqref{Sc-est-5},  and Young's inequality (i.e. $ab \leq p^{-1}a^p + q^{-1}b^q$, $a,b\geq0$ with $p=\frac{\ell-1}{m}$ and $q=\frac{\ell-1}{\ell-m-1}$) to deduce that
\begin{equation*}
\bigl\|\Att^0(t,\wtt)[\del{x}^\ell,\Btt(t,\wtt)]V\bigr\|_{L^2} \lesssim  t^{\ell(1-3\sspeed)}\bigl( 1+ \norm{\wtt}_{W^{1,\infty}}^{\ell-1}\bigr)\bigl(\norm{\wtt}_{W^{1,\infty}}\norm{V}_{H^{\ell-1}}+\norm{\wtt}_{H^\ell}\norm{V}_{L^\infty}\bigr).
\end{equation*}
\end{proof}

\section{Existence and asymptotics\label{sec:exist-asymp}}
We are now in a position to establish the existence of $\Tbb^2$-symmetric solutions to the relativistic Euler equations \eqref{relEulA} globally to the future by solving the Fuchsian global initial value problem (GIVP)
\begin{align}
\del{t}\Wtt + \Btt(t,\wtt)\del{x}\Wtt &= \frac{3\sspeed-1}{t}\pip \Wtt \hspace{0.5cm} \text{in $(0,1]\times\Tbb$,}
\label{T2-Eul-H.1}\\
\Wtt&= (\ztt_0,\wtt_0)^{\tr} \hspace{0.95cm} \text{in $\{1\}\times \Tbb$,} 
\label{T2-Eul-H.2}
\end{align}
under a suitable small initial data assumption and for sound speeds satisfying $1/3 < \sspeed<1$.

\begin{rem}\label{rem-idata}
The trivial initial data 
\begin{equation*}
\Wtt|_{t=1}=(\ztt_0,\wtt_0)=(0,0)
\end{equation*}
corresponds, by \eqref{conformal}, \eqref{vi-def}-\eqref{v0-def}, \eqref{zttt-def}-\eqref{vttt-def}, \eqref{cov}, \eqref{phi-def} and \eqref{Wtt-def}, to
the initial data
\begin{equation*}
(\rho,\vt^i)|_{t=1} = \bigl(\rho_c,-\delta^i_0\bigr)
\end{equation*}
for the relativistic Euler equations, which is the initial data that uniquely generates the family \eqref{Hom-A} of spatially homogeneous and isotropic solutions of the relativistic Euler equations. 
Since solutions $\Wtt=(\ztt,\wtt)$ of the GIVP \eqref{T2-Eul-H.1}-\eqref{T2-Eul-H.2} uniquely determine solutions $(\rho,\vt^i)$ of the relativistic Euler equations via
\eqref{conformal}, \eqref{vi-def}-\eqref{v0-def}, \eqref{zttt-def}-\eqref{vttt-def}, \eqref{cov}, \eqref{phi-def} and \eqref{Wtt-def}, it follows that solutions to the Fuchsian GIVP \eqref{T2-Eul-H.1}-\eqref{T2-Eul-H.2} that are generated by small initial data correspond to nonlinear $\Tbb^2$-symmetric perturbations of the homogeneous and isotropic solutions \eqref{Hom-A} that exist globally to the future. The existence of such perturbed solutions is established in the following theorem.     
\end{rem}

\begin{thm}\label{thm-exist}
Suppose $k\in \Zbb_{>\frac{3}{2}}$ and  $\frac{1}{3}<\sspeed <\frac{k+1}{3k}$.
Then there exists a $\delta_0>0$ such that for all $\delta\in (0,\delta_0)$, and initial data $\Wtt(1)=(\ztt_0,\wtt_0)^{\tr}\in H^{k}(\Tbb,\Rbb^2)$ satisfying
\begin{equation*}
\norm{\Wtt(1)}_{H^k} = \sqrt{\norm{\ztt_0}_{H^k}^2 + \norm{\wtt_0}_{H^k}^2} \leq \delta, 
\end{equation*} 
there exists a unique solution 
\begin{equation*}
\Wtt=(\ztt, \wtt )^{\tr} \in C^0\bigl((0,1],H^{k}(\Tbb,\Rbb^2)\bigr)\cap C^1 \bigl((0,1],H^{k-1}(\Tbb,\Rbb^2)\bigr)
\end{equation*}
to the Fuchsian GIVP \eqref{T2-Eul-H.1}-\eqref{T2-Eul-H.2} on $(0,1]\times \Tbb$ and a constant $C=C(\delta_0)$ such that $\Wtt$ is bounded by
\begin{equation*}
\norm{\Wtt(t)}_{C^{k-1,\frac{1}{2}}} \lesssim \norm{\Wtt(t)}_{H^k} \leq C\delta
\end{equation*}
for all $t\in (0,1]$.
Moreover, the solution $\Wtt$ satisfies the following properties:
\begin{enumerate}[(a)]
\item There exist functions $z_*,\wtt_* \in H^{k-1}(\Tbb)\subset C^{k-2,\frac{1}{2}}(\Tbb)$ such that $z=t^{1-3\sspeed}\ztt$ and $\wtt$ converge in
$H^{k-1}(\Tbb)$ as $t\searrow 0$ to  $z_*$ and $\wtt_*$, respectively, and there exists a constant $C=C(\delta_0)$ such that decay estimates
\begin{align*}
\norm{\wtt(t)-\wtt_*}_{ C^{k-2,\frac{1}{2}}} \lesssim \norm{\wtt(t)-\wtt_*}_{H^{k-1}}  &\leq C\delta t^{(k-1)(1-3\sspeed)+1}
\intertext{and}
\norm{z(t)-z_*}_{C^{k-2,\frac{1}{2}}}\lesssim \norm{z(t)-z_*}_{H^{k-1}}  &\leq C\delta t^{k(1-3\sspeed)+1}
\end{align*}
hold for all $t\in (0,1]$. 
\item The pair $(z,\wtt)$,  where $z=t^{1-3\sspeed}\ztt$, determines a unique, $\Tbb^2$-symmetric $C^1$-solution $(\rho,\vt^i)$ of the relativistic Euler equations \eqref{relEulA} on $(0,1]\times \Tbb^3$ via 
\begin{equation*}
(\rho,\vt^i) = \left( \frac{\dsp t^{3(1+\sspeed)}\rho_c e^{(1+\sspeed) z}}{ \dsp\Bigl(1+\Bigl[\phi^{-1}\bigl(t^{1-3\sspeed}\wtt
\bigr)\Bigr]^2\Bigr)^{\frac{1+\sspeed}{2}}},
 t\biggl(-\delta^i_0  \sqrt{1+\Bigl[\phi^{-1}\bigl(t^{1-3\sspeed}\wtt
\bigr)\Bigr]^2}+ \delta^i_1 \phi^{-1}\bigl(t^{1-3\sspeed}\wtt
\bigr) \biggr)\right)
\end{equation*}
where $\phi\in C^\infty(\Rbb,\Rbb)$ is the diffeomorphism defined by \eqref{phi-def}.
\item Let
\begin{equation*}
\Wsc_*^+ = \wtt_*^{-1}((0,\infty)) \AND \Wsc_*^- = \wtt_*^{-1}((-\infty,0)).
\end{equation*}
Then for every $x_{\pm}\in \Wsc^*_{\pm}$, there exists a time $t_{\pm} \in (0,1]$ such that
\begin{align*}
\rho(t,x_{\pm}) & = t^{\frac{2(1+\sspeed)}{1-\sspeed}} \biggl(\frac{\dsp\rho_c e^{(1+\sspeed)z_*(x_\pm)}}{|w_*(x_\pm)|^{\frac{1+\sspeed}{1-\sspeed}} }+\Ord\bigl(t^{\lambdat}\bigr)\biggr),   
\\
\vt^i(t,x_{\pm}) &=t^{\frac{2(1-2\sspeed)}{1-\sspeed}}\Bigl(|w_*(x_\pm)|^{\frac{1}{1-\sspeed}}\bigl(-\delta^i_0 \pm \delta^i_1\bigr) +\Ord\bigl(t^\lambda\bigr)\Bigr), 
\end{align*}
for all $t\in (0,t_{\pm}]$, where
\begin{equation*}
\lambdat = \min\biggl\{ k(1-3\sspeed)+1, \frac{2(3\sspeed-1)}{1-\sspeed}\biggr\} \AND
\lambda = \min\biggl\{ (k-1)(1-3\sspeed)+1, \frac{2(3\sspeed-1)}{1-\sspeed}\biggr\}, 
\end{equation*}
and moreover,
\begin{align*}
 \lim_{t\searrow 0} t^{\frac{2(1+\sspeed)}{\sspeed-1}} \rho(t,x_{\pm}) & = \rho_c |w_*(x_\pm)|^{\frac{1+\sspeed}{\sspeed-1}} e^{(1+\sspeed)z_*(x_\pm)},
\\
\lim_{t\searrow 0} t^{\frac{2(1-2\sspeed)}{\sspeed-1}}\vt^i(t,x_{\pm}) &=|w_*(x_\pm)|^{\frac{1}{1-\sspeed}} \bigl(-\delta^i_0 \pm \delta^i_1 \bigr). 
\end{align*}
Furthermore, if $k\in \Zbb_{>\frac{5}{2}}$, then 
\begin{equation*}
\frac{\del{x}\rho(t,x_{\pm})}{\rho(t,x_\pm)} =(1+\sspeed)\Biggl(\mp\frac{\del{x}\wtt_*(x_\pm)}{(1-\sspeed)|\wtt_*(x_\pm)|} + \del{x}z_*(x_\pm)+\Ord\bigl(t^{\lambdat}\bigr)\Biggr) 
\end{equation*}
for all $t\in (0,t_{\pm}]$, and moreover,
\begin{equation*} 
\lim_{t\searrow 0} \frac{\del{x}\rho(t,x_{\pm})}{\rho(t,x_\pm)} =(1+\sspeed)\biggl(\mp\frac{\del{x}\wtt_*(x_\pm)}{(1-\sspeed)|\wtt_*(x_\pm)|} + \del{x}z_*(x_\pm)\biggr).
\end{equation*} 
\end{enumerate}
\end{thm}
\begin{proof}
$\;$

\noindent\underline{Existence:}
Suppose $k\in \Zbb_{>\frac{3}{2}}$, $\delta>0$, $\frac{1}{3}<\sspeed <\frac{k+1}{3k}$ and choose initial data 
initial data $\ztt_0,\wtt_0\in H^{k}(\Tbb)$ satisfying 
\begin{equation} \label{exist-idata}
\norm{\Wtt(1)}_{H^k} = \sqrt{\norm{\ztt_0}_{H^k}^2 + \norm{\wtt_0}_{H^k}^2} < \delta.
\end{equation} 
By the results of Sections \ref{sec:cov} and \ref{sec:rescale}, we know that the matrices  $\Att^0(t,\wtt)$ and $\Btt(t,\wtt)$ defined by
\eqref{AcBc-def} are smooth for $(t,\wtt)\in (0,1]\times \Rbb$,  $\Att^0(t,\wtt)$ and
$\Att^1(t,\wtt)=\Att^0(t,\wtt)\Btt(t,\wtt)$ are symmetric for all $(t,\wtt)\in (0,1]\times \Rbb$, and
$\Att^0(t,\wtt)$ is positive definite. As a consequence, the system
\eqref{T2-Eul-H.1} is symmetrisable and multiplying it on the left by $\Att^0(t,\wtt)$ yields the symmetric hyperbolic system \eqref{T2-Eul-F}.
We therefore can appeal to
standard local-in-time existence and uniqueness theorems for symmetric hyperbolic systems, e.g. \cite[Thm.~10.1]{BenzoniSerre:2007} or \cite[Thm.~2.1]{Majda:1984},  to conclude the existence of
a unique solution 
\begin{equation}\label{exist-sol-A}
\Wtt=(\ztt, \wtt )^{\tr} \in C^0\bigl((t_*,1],H^{k}(\Tbb,\Rbb^2)\bigr)\cap C^1 \bigl((t^*,1],H^{k-1}(\Tbb,\Rbb^2)\bigr)
\end{equation}
to the IVP \eqref{T2-Eul-H.1}-\eqref{T2-Eul-H.2} on a time interval $(t_*,1]$ for some $t_*\in [0,1)$. Moreover,
by the continuation principle for symmetric hyperbolic equations, e.g. \cite[Thm.~10.3]{BenzoniSerre:2007} or  \cite[Thm.~2.2]{Majda:1984}, if $t_* \in (0,1]$ and 
$\sup_{t\in (t_*,1]}\norm{\Wtt(t)}_{W^{1,\infty}} < \infty$, 
then there exists a $t^*\in [0,t_*)$ such that the solution \eqref{exist-sol-A} can be uniquely continued, as a solution of the same regularity, to the larger time interval $(t^*,1]$. 

In light of the continuation principle,  the solution \eqref{exist-sol-A} will exist on the time
interval $(0,1]$ provided we can show, by choosing $\delta_0>0$ sufficiently small, that the solution
satisfies a uniform bound of the form
\begin{equation}\label{exist-cont-A}
\sup_{t\in (t_*,1]}\norm{\Wtt(t)}_{W^{1,\infty}} \leq C
\end{equation}
for some constant $C>0$ independent of $t_*$. We will show that such a bound exists using energy estimates. To derive these estimates, we assume that $\ell \in \Zbb$ satisfies $0\leq \ell \leq k$ and apply $\Att^0(t,\wtt)\del{x}^\ell$
to \eqref{T2-Eul-H.1} to get
\begin{equation*}
\Att^0(t,\wtt)\del{t}\del{x}^\ell \Wtt + \Att^1(t,\wtt)\del{x}\del{x}^\ell\Wtt  = \frac{3\sspeed-1}{t}\Att^0(t,\wtt)\pip \del{x}^\ell\Wtt - \Ftt_\ell
\end{equation*} 
where 
\begin{equation}\label{Ftt-ell-def}
\Ftt_\ell = \Att^0(t,\wtt)[\del{x}^\ell,\Btt(t,\wtt)]\del{x}\Wtt.
\end{equation}
Multiplying the above equations on the left by $\del{x}^\ell \Wtt^{\tr}$ and integrating by parts,
we obtain the energy identity
\begin{equation}\label{energy-est-B}
-\frac{1}{2}\del{t}\nnorm{\del{x}^\ell \Wtt}_0^2 =-\frac{1}{2}\ip{\del{x}^\ell\Wtt}{\Div\!\Att(t,\wtt)\,\del{x}^\ell \Wtt} -\frac{3\sspeed-1}{t}\nnorm{\pip \del{x}^\ell\Wtt}_0^2+\ip{\Wtt}{\Ftt_\ell}
\end{equation}
where $\Div\!\Att(t,\wtt)$ is the matrix divergence defined previously by \eqref{DivAc-def}, c.f.
\eqref{DivAc-on-shell},
\begin{equation}\label{energy-norm}
\nnorm{(\cdot)}_0^2 := \ip{(\cdot)}{\Att^0(t,\wtt)(\cdot)}
\end{equation}
is the \textit{energy norm}, and in deriving the middle term on the right hand side of \eqref{energy-est-B}, we
have employed the identity
\begin{equation}\label{linear-id}
\ip{\del{x}^\ell \Wtt}{\Att^0(t,\wtt)\pip \del{x}^\ell \Wtt} = \nnorm{\pip\del{x}^\ell\Wtt}_0^2,
\end{equation}
which follows easily from \eqref{pip-Ac0-com}. We also define a higher order energy norm via
\begin{equation}\label{hi-energy-norm}
\nnorm{\Wtt}_k^2= \sum_{\ell=0}^k\nnorm{\del{x}^\ell \Wtt}_0^2,
\end{equation}
and note, by Lemma \ref{lem-Ac0-lbnd} and Sobolev's inequality (Theorem \ref{Sobolev}), that
\begin{equation} \label{norm-bnd}
\norm{\Wtt}_{W^{1,\infty}} \lesssim \sqrt{\gamma} \norm{\Wtt}_{H^k} \leq \nnorm{\Wtt}_k.
\end{equation}

To proceed, we estimate the terms on the right hand side of \eqref{energy-est-B} by noting from \eqref{DivAc-def}, \eqref{DwAc1-rep} and \eqref{linear-id} that 
 $\ip{\del{x}^\ell\Wtt}{\Div\!\Att(t,\wtt)\,\del{x}^\ell \Wtt}$ can be expressed as
\begin{align*}
\ip{\del{x}^\ell\Wtt}{\Div\!\Att(t,\wtt)\,\del{x}^\ell \Wtt}
=&\Bigl( t^{1-3\sspeed}(a_1\del{x}\ztt+a_2\del{x}\wtt)+t^{-1}(-2(3\sspeed-1)+a_3)\Bigr) \nnorm{\pip\del{x}^\ell \Wtt}_0^2
\\
&\qquad + t^{1-3\sspeed}\del{x}\wtt 
\bigl\langle
(\Att^0(t,\wtt))^{\frac{1}{2}}\del{x}^\ell \Wtt \bigl|
Q(\Att^0(t,\wtt))^{\frac{1}{2}}\del{x}^\ell \Wtt\bigr\rangle.
\end{align*}
From this, we obtain
\begin{align*}
-\frac{1}{2}\ip{\del{x}^\ell\Wtt}{\Div\!\Att(t,\wtt)\,\del{x}^\ell \Wtt} -\frac{3\sspeed-1}{t}\nnorm{\pip \del{x}^\ell\Wtt}_0^2 =&-\frac{1}{2}\Bigl( t^{1-3\sspeed}(a_1\del{x}\ztt+a_2\del{x}\wtt)+t^{-1}a_3\Bigr) \nnorm{\pip\del{x}^\ell \Wtt}_0^2
\\
&\quad -\frac{1}{2}t^{1-3\sspeed}\del{x}\wtt 
\bigl\langle
(\Att^0(t,\wtt))^{\frac{1}{2}}\del{x}^\ell \Wtt \bigl|
Q(\Att^0(t,\wtt))^{\frac{1}{2}}\del{x}^\ell \Wtt\bigr\rangle,
\end{align*}
which we observe, with the help of Lemma \ref{lem-a1a2a3Q}, is bounded above by
\begin{equation} \label{energy-est-C}
-\frac{1}{2}\ip{\del{x}^\ell\Wtt}{\Div\!\Att(t,\wtt)\,\del{x}^\ell \Wtt} -\frac{3\sspeed-1}{t}\nnorm{\pip \del{x}^\ell\Wtt}_0^2 \lesssim t^{1-3\sspeed}\norm{\Wtt}_{W^{1,\infty}}\nnorm{\del{x}^\ell \Wtt}_0^2
\lesssim  t^{1-3\sspeed}\nnorm{\Wtt}_k^3,
\end{equation}
where in deriving the second inequality we employed \eqref{norm-bnd}. We further observe from
the definition \eqref{Ftt-ell-def} of $\Ftt_\ell$, the positivity of $3\sspeed-1$, the Cauchy-Schwartz inequality, the commutator estimate from Lemma \ref{comm-lem}, and \eqref{norm-bnd} that the term $\ip{\Wtt}{\Ftt_\ell}$ can be bounded above by
\begin{align}
\ip{\Wtt}{\Ftt_\ell} \leq \norm{\Wtt}_{L^2}\norm{\Ftt_\ell}_{L^2} &\lesssim t^{\ell(1-3\sspeed)}
\norm{\Wtt}_{L^2}\bigl(1+\norm{\wtt}^{\ell-1}_{W^{1,\infty}}\bigr)\bigl(\norm{\wtt}_{W^{1,\infty}}\norm{\del{x}\Wtt}_{H^{\ell-1}}+ \norm{\wtt}_{H^\ell}\norm{\del{x}\Wtt}_{L^\infty}\bigr)\notag \\
& \lesssim t^{k(1-3\sspeed)}\bigl(1+\nnorm{\Wtt}^{k-1}_{k}\bigr)\nnorm{\Wtt}^3_k. 
\label{energy-est-D}
\end{align}
Then summing \eqref{energy-est-B} over $\ell$ from $\ell=0$ to $\ell=k$, while using \eqref{energy-est-C}
and \eqref{energy-est-D} to estimate the resulting terms on the right hand side, yields the energy inequality
\begin{equation}\label{energy-est-E}
-\frac{1}{2}\del{t}\nnorm{\Wtt}_k^2 \lesssim  t^{k(1-3\sspeed)}\bigl(1+\nnorm{\Wtt}^{k-1}_{k}\bigr)\nnorm{\Wtt}^3_k,
\end{equation}
which we can express as
\begin{equation}\label{energy-est-F}
-\del{t}\nnorm{\Wtt}_k \lesssim  t^{k(1-3\sspeed)}\bigl(1+\nnorm{\Wtt}^{k-1}_{k}\bigr)\nnorm{\Wtt}_k^2
\end{equation}
by dividing through by $\nnorm{\Wtt}_{k}$.

Now, at $t=1$, it is clear from \eqref{energy-norm}, \eqref{hi-energy-norm}, Sobolev's inequality (Theorem \ref{Sobolev}), the 
initial data bound \eqref{exist-idata}, and Lemma \ref{lem-Ac0-lbnd} that there exists
a constant $C_0(\delta)>0$ such that
\begin{equation} \label{exist-idata-B}
\nnorm{\Wtt(1)}_k \leq C_0(\delta)\norm{\Wtt(1)}_{H^k}< C_0(\delta)\delta.
\end{equation}
Next, we fix $\delta_0>0$ and assume that $\delta\in (0,\delta_0/2]$. We then let $\ttld \in (t_*,1]$ be the first time for which 
$\nnorm{\Wtt(\ttld)}_k =C_0(\delta_0)\delta_0$
and if there is no such time then we set $\ttld=t_*$. Then, by \eqref{energy-est-F}, we have
\begin{equation*}
-\del{t}\nnorm{\Wtt}_k \leq C_1(\delta_0)\delta_0 t^{k(1-3\sspeed)}\nnorm{\Wtt}_k, \quad 0<t_*\leq \ttld <t\leq 1,
\end{equation*}
for some constant $C_1(\delta_0)>0$.  Since $k(1-3\sspeed)>-1$ by assumption, the above energy inequality and \eqref{exist-idata-B} allow us to conclude via an application of Gron\"{o}wall's inequality that 
\begin{equation} \label{energy-est-G}
\nnorm{\Wtt(t)}_k \leq e^{C_1(\delta_0)\delta_0\beta(t)}\nnorm{\Wtt(1)}_k < e^{C_1(\delta_0)\delta_0\bar{\beta}}C_0(\delta)\delta,  \quad 0<t_*\leq \ttld <t\leq 1,
\end{equation}
where
\begin{equation*}
\beta(t) = \frac{1-t^{k(1-3\sspeed)+1}}{k(1-3\sspeed)+1} \leq \bar{\beta}:= \frac{1}{k(1-3\sspeed)+1}.
\end{equation*}
But $\lim_{\delta\searrow 0} e^{C_1(\delta_0)\delta_0\bar{\beta}}C_0(\delta)\delta=0$, and so we can choose $\delta \in (0,\delta_0/2]$ small enough to ensure that
$e^{C_1(\delta_0)\delta_0\bar{\beta}}C_0(\delta)\delta< C_0(\delta_0)\delta_0$.
Doing so, we conclude, by \eqref{norm-bnd} and \eqref{energy-est-G},  that
$\nnorm{\Wtt(t)}_k \leq   C_0(\delta_0)\delta_0$ for $0<t_*\leq \ttld <t\leq 1$, and hence, that 
$\ttld=t_*$. By \eqref{norm-bnd},  we therefore have established that the bound \eqref{exist-cont-A}
holds for some constant $C>0$ independent of $t_*$ provided $\delta>0$ is chosen small enough. 
Thus, by the continuation principle, we have that 
\begin{equation} \label{t*=0}
t_*=0, 
\end{equation}
and consequently, the solution \eqref{exist-sol-A}
exists on the time interval $(0,1]$ and determines a unique solution of the Fuchsian GIVP \eqref{T2-Eul-H.1}-\eqref{T2-Eul-H.2}. Moreover, the above arguments show that this solution is bounded by
\begin{equation} \label{sol-bnd}
\norm{\Wtt(t)}_{W^{1,\infty}}\lesssim \norm{\Wtt(t)}_{H^k} \leq C(\delta_0)\norm{\Wtt(1)}_{H^k},
\quad 0<t\leq 1.
\end{equation}

\bigskip

\noindent\underline{Decay estimates:}
Since $\pi \pi^\perp =0$ by \eqref{pip-Ac0-com}, applying the projection $\pi$ to \eqref{T2-Eul-H.1} yields
\begin{equation} \label{pi-Eul}
\del{t}\pi\Wtt + \pi\Btt(t,\wtt)\del{x}\Wtt = 0,
\end{equation}
and so integrating this expression between $t_1$ and $t$, where $0<t_1<t\leq 1$, yields
\begin{equation} \label{asymp-A}
(0,\wtt(t))^{\tr}-(0,\wtt(t_1))^{\tr}\oset{\eqref{Wtt-def}}{=}\pi\Wtt(t_1)-\pi \Wtt(t) =- \int_{t_1}^t  \pi\Btt(\tau,\wtt(\tau))\del{x}\Wtt(\tau)\, d\tau.
\end{equation}
Because $k-1>1/2$, the product estimate from Theorem \ref{Product}.(ii)  allows us to estimate the $\Btt(t,\wtt(t))\del{x}\Wtt(t)$ by
\begin{equation}\label{Btt-dWtt-est-A}
\norm{\Btt(t,\wtt(t))\del{x}\Wtt(t) }_{H^{k-1}}\lesssim  \norm{\Btt(t,\wtt(t))}_{H^{k-1}}
\norm{\del{x}\Wtt(t)}_{H^{k-1}}.
\end{equation}
On the other hand, by Corollary \ref{cor-Ac0Bc} and Moser's estimate (Theorem \ref{Moser}), we have that
\begin{equation}\label{Btt-dWtt-est-B}
 \norm{\Btt(t,\wtt(t))}_{H^{k-1}} \lesssim t^{(k-1)(1-3 \sspeed)} (1+\norm{\wtt(t)}^{k-2}_{L^\infty})\bigl(1+\norm{\wtt(t)}_{H^{k-1}}\bigr).
\end{equation}
Employing these estimates in conjuntion with \eqref{sol-bnd} gives
\begin{equation} \label{asymp-B}
\norm{\Btt(t,\wtt(\tau))\del{x}\Wtt(t) }_{H^{k-1}}\leq  t^{(k-1)(1-3\sspeed)}C(\delta_0)\norm{\Wtt(1)}_{H^k}.
\end{equation}
With the help of this estimate, we can then use \eqref{asymp-A} to bound the difference
$\wtt(t)-\wtt(t_1)$ by  
\begin{equation} \label{wtt-asymp-est}
\norm{\wtt(t)-\wtt(t_1)}_{H^{k-1}} \leq C(\delta_0)\int^t_{t_1} 
 \tau^{(k-1)(1-3\sspeed)}\, d\tau \leq C(\delta_0)\norm{\Wtt(1)}_{H^k}\bigl( t^{(k-1)(1-3\sspeed)+1}- t_1^{(k-1)(1-3\sspeed)+1} \bigr)
\end{equation} 
for  $0<t_1\leq t \leq 1$. Since $(k-1)(1-3\sspeed)+1>0$ by assumption, we conclude from this estimate and Sobolev's inequality (Theorem \ref{Sobolev})  that the limit $\lim_{t\searrow 0}\wtt(t)$, denoted $\wtt_*$, exists in $H^{k-1}(\Tbb)\subset C^{k-2,\frac{1}{2}}(\Tbb)$, and by
sending $t_1\searrow 0$, that
\begin{equation}\label{asymp-C}
\norm{\wtt(t)-\wtt_*}_{ C^{k-2,\frac{1}{2}}}\lesssim \norm{\wtt(t)-\wtt_*}_{H^{k-1}}  \leq C(\delta_0)\norm{\Wtt(1)}_{H^k} t^{(k-1)(1-3\sspeed)+1}, \quad 0<t\leq 1.
\end{equation}

Next, multiplying \eqref{T2-Eul-H.1} by $t^{1-3\sigma}\pip$, we find after a short calculation that
\begin{equation} \label{pip-Eul}
\del{t}(t^{1-3\sspeed}\pip\Wtt) + t^{1-3\sspeed}\pip\Btt(t,\wtt)\del{x}\Wtt = 0.
\end{equation}
Integrating this from $t_1$ to $t$ gives
\begin{equation*}
(z(t),0)^{\tr}-(z(t_1),0)^{\tr} \oset{\eqref{Wtt-def}}{=} t^{1-3\sspeed}\pip\Wtt(t)-  
t_1^{1-3\sspeed}\pip\Wtt(t_1) = -\int_{t_1}^t  \tau^{1-3\sspeed}\pip\Btt(\tau,\wtt(\tau))\del{x}\Wtt(\tau)\, d\tau,
\end{equation*}
which by \eqref{asymp-B}, we use to bound the difference $z(t)-z(t_1)$  by 
\begin{equation} \label{z-asymp-est}
\norm{z(t)-z(t_1)}_{H^{k-1}}\leq C(\delta_0)\int^t_{t_1} 
 \tau^{k(1-3\sspeed)}\, d\tau \leq C(\delta_0)\norm{\Wtt(1)}_{H^k}\bigl( t^{k(1-3\sspeed)+1}- t_1^{k(1-3\sspeed)+1} \bigr)
\end{equation}
for  $0<t_1\leq t \leq 1$.  Since $k(1-3\sspeed)+1>0$ by assumption, the same argument that lead to \eqref{asymp-C} implies the existence of a $z_* \in H^{k-1}(\Tbb)\subset C^{k-1,\frac{1}{2}}(\Tbb)$ such that
\begin{equation}\label{asymp-D}
\norm{z(t)-z_*}_{ C^{k-2,\frac{1}{2}}}\lesssim\norm{z(t)-z_*}_{H^{k-1}}  \leq C(\delta_0)\norm{\Wtt(1)}_{H^k} t^{k(1-3\sspeed)+1}, \quad 0<t\leq 1.
\end{equation}

\bigskip

\noindent\underline{$\Tbb^2$-symmetric solutions of the relativistic Euler equations:}
From the definitions \eqref{vi-def}-\eqref{zetat-def}, \eqref{zttt-def}-\eqref{vttt-def}, \eqref{cov}, and \eqref{Wtt-def}, it is straightforward, with the help of \eqref{conformal}, \eqref{v0-def}, \eqref{exist-sol-A},  \eqref{t*=0} and Sobolev's inequality, to 
verify that the pair  $(z,\wtt)$ determines a $\Tbb^2$-symmetric, $C^1$ solution $(\rho,\vt^i)$ of the relativistic Euler equations \eqref{relEulA} on $(0,1]\times \Tbb^3$ via
\begin{equation} \label{T2-sol}
(\rho,\vt^i) = \left( \frac{\dsp t^{3(1+\sspeed)}\rho_c e^{(1+\sspeed) z}}{ \dsp\Bigl(1+\Bigl[\phi^{-1}\bigl(t^{1-3\sspeed}\wtt
\bigr)\Bigr]^2\Bigr)^{\frac{1+\sspeed}{2}}},
 t\biggl(-\delta^i_0  \sqrt{1+\Bigl[\phi^{-1}\bigl(t^{1-3\sspeed}\wtt
\bigr)\Bigr]^2}+ \delta^i_1 \phi^{-1}\bigl(t^{1-3\sspeed}\wtt
\bigr) \biggr)\right),
\end{equation}
where $\phi\in C^\infty(\Rbb,\Rbb)$ is the diffeomorphism (see Lemma \ref{phi-lem}) defined by \eqref{phi-def}.

\bigskip
\noindent\underline{Solution asymptotics:}
For a given continuous limit function $\wtt_*\in C^{k-2,\frac{1}{2}}(\Tbb)\subset H^{k-1}(\Tbb)$ (recall $k\geq 2$ by assumption), we define open sets in $\Tbb$ by
\begin{equation*}
\Wsc_*^+ = \wtt_*^{-1}((0,\infty)) \AND \Wsc_*^- = \wtt_*^{-1}((-\infty,0)).
\end{equation*}
Then for any $x_{\pm} \in \Wsc_*^{\pm}$, it follows from the decay estimate \eqref{asymp-C} that
\begin{equation}
t^{1-3\sspeed}\wtt(t,x_{\pm})= t^{1-3\sspeed}\bigl(\wtt_*(x_\pm) + \Ord\bigl(t^{(k-1)(1-3\sspeed)+1}\bigr)\bigr),
\label{asymp-exp-A}
\end{equation}
which allows us to deduce from the assumption  $k(1-3\sspeed)+1>0$ and the asymptotic behaviour of $\phi^{-1}$, see Lemma \ref{phi-lem}, that 
\begin{equation*}
\phi^{-1}\bigl(t^{1-3\sspeed}\wtt(t,x_{\pm})\bigr) = \pm t^{\frac{1-3\sspeed}{1-\sspeed}}\Bigl(|w_*(x_\pm)|^\frac{1}{1-\sspeed}+ \Ord\Bigl(t^{(k-1)(1-3\sspeed)+1}\Bigr)\Bigr)\Bigl(1+\Ord\Bigl(t^{\frac{2(3\sspeed-1)}{1-\sspeed}}\Bigr)\Bigr),  \quad 0<t<t_{\pm},
\end{equation*}
for $t_{\pm}=t_{\pm}(x_{\pm})$ sufficiently small, which letting
\begin{equation*}
\lambda = \min\biggl\{ (k-1)(1-3\sspeed)+1, \frac{2(3\sspeed-1)}{1-\sspeed}\biggr\}>0, 
\end{equation*}
we can express more compactly as 
\begin{equation}
\phi^{-1}\bigl(t^{1-3\sspeed}\wtt(t,x_{\pm})\bigr) = \pm t^{\frac{1-3\sspeed}{1-\sspeed}}\bigl(|w_*(x_\pm)|^\frac{1}{1-\sspeed}+ \Ord(t^\lambda)\bigr),  \quad 0<t<t_{\pm}. \label{asymp-exp-B} 
\end{equation}
Plugging this into \eqref{T2-sol} yields, with the help of the decay estimate \eqref{asymp-D}, the asymptotic expansions
\begin{align*}
\rho(t,x_{\pm}) & = \frac{\dsp t^{3(1+\sspeed)}\rho_c e^{(1+\sspeed)\bigl(z_*(x_\pm)+
\Ord (t^{k(1-3\sspeed)+1})  \bigr)}}{\dsp \Bigl(1+ t^{\frac{2(1-3\sspeed)}{1-\sspeed}}\Bigl[|w_*(x_\pm)|^{\frac{1}{1-\sspeed}} +\Ord(t^\lambda)\Bigr]^2 \Bigr)^{\frac{1+\sspeed}{2}}},  
\\
\vt^i(t,x_{\pm}) &=t\biggl(-\delta^i_0  \sqrt{1+t^{\frac{2(1-3\sspeed)}{1-\sspeed}}\Bigl[|w_*(x_\pm)|^{\frac{1}{1-\sspeed}} +\Ord(t^\lambda)\Bigr]^2 }+ \delta^i_1 t^{\frac{(1-3\sspeed)}{1-\sspeed}}\Bigl(\pm |w_*(x_\pm)|^{\frac{1}{1-\sspeed}} +\Ord(t^\lambda) \Bigr) \biggr),
\end{align*}
which hold for $t\in (0,t_\pm)$ and we note can be simplified to
\begin{align*}
\rho(t,x_{\pm}) & = t^{\frac{2(1+\sspeed)}{1-\sspeed}} \biggl(\frac{\dsp\rho_c e^{(1+\sspeed)z_*(x_\pm)}}{|w_*(x_\pm)|^{\frac{1+\sspeed}{1-\sspeed}} }+\Ord(t^{\lambdat})\biggr),   
\\
\vt^i(t,x_{\pm}) &=t^{\frac{2(1-2\sspeed)}{1-\sspeed}}\Bigl(|w_*(x_\pm)|^{\frac{1}{1-\sspeed}}\bigl(-\delta^i_0 \pm \delta^i_1\bigr) +\Ord(t^\lambda)\Bigr), 
\end{align*}
where
\begin{equation*}
\lambdat = \min\biggl\{ k(1-3\sspeed)+1, \frac{2(3\sspeed-1)}{1-\sspeed}\biggr\}>0. 
\end{equation*}
Letting $t\searrow 0$ in these expressions yields
\begin{align*}
 \lim_{t\searrow 0} t^{\frac{2(1+\sspeed)}{\sspeed-1}} \rho(t,x_{\pm}) & = \rho_c |w_*(x_\pm)|^{\frac{1+\sspeed}{\sspeed-1}} e^{(1+\sspeed)z_*(x_\pm)},
\\
\lim_{t\searrow 0} t^{\frac{2(1-2\sspeed)}{\sspeed-1}}\vt^i(t,x_{\pm}) &=|w_*(x_\pm)|^{\frac{1}{1-\sspeed}} \bigl(-\delta^i_0 \pm \delta^i_1 \bigr).
\end{align*}

\bigskip

\noindent \underline{Fractional density gradient asymptotics:}
A short calculation using \eqref{T2-sol} shows that the fractional density gradient is given by
\begin{equation} \label{asymp-E}
\frac{\del{x}\rho}{\rho} =  \frac{\dsp (1+\sspeed)\biggl(- \phi^{-1}\bigl(t^{1-3\sspeed}\wtt
\bigr) \bigl(\phi^{-1}\bigr)'\bigl(t^{1-3\sspeed}\wtt
\bigr)t^{1-3\sspeed}\del{x}\wtt + 
\Bigl(1+\Bigl[\phi^{-1}\bigl(t^{1-3\sspeed}\wtt
\bigr)\Bigr]^2\Bigl)\del{x}z \biggr)}{\dsp 1+\Bigl[\phi^{-1}\bigl(t^{1-3\sspeed}\wtt
\bigr)\Bigr]^2}.
\end{equation}
Also, from the asymptotic equation \eqref{asymp-exp-A},  and the  asymptotic behaviour of the map $\phi^{-1}(w)$ and its derivative $(\phi^{-1})'(w)$,  see Lemma  \ref{phi-lem}, we note that
\begin{align*}
\phi^{-1}\bigl(t^{1-3\sspeed}\wtt(t,x_\pm)
\bigr) \bigl(\phi^{-1}\bigr)'\bigl(t^{1-3\sspeed}\wtt(t,x_\pm)
\bigr) &= \pm\bigl|t^{1-3\sspeed}\wtt(t,x_\pm)\bigr|^{\frac{1+\sspeed}{1-\sspeed}}
\biggl(\frac{1}{1-\sspeed}+ \Ord\Bigl(\bigl|t^{1-3\sspeed}\wtt(t,x_\pm)\Bigr|^{\frac{2}{\sspeed-1}}\Bigr)\biggr)\\
= \pm t^{\frac{(1-3\sspeed)(1+\sspeed)}{1-\sspeed}}&\Bigl|\wtt_*(x_\pm)+ \Ord\bigl(t^{(k-1)(1-3\sspeed)+1}\bigr)\Bigr|^{\frac{1+\sspeed}{1-\sspeed}}
\biggl(\frac{1}{1-\sspeed}+ \Ord\Bigl(t^{\frac{2(3\sspeed -1)}{1-\sspeed}}\Bigr)\biggr)
\end{align*}
for $t\in (0,t_\pm]$, which can be simplified to
\begin{equation} \label{asymp-exp-D}
\phi^{-1}\bigl(t^{1-3\sspeed}\wtt(t,x_\pm)
\bigr) \bigl(\phi^{-1}\bigr)'\bigl(t^{1-3\sspeed}\wtt(t,x_\pm)
\bigr) = \pm t^{\frac{(1-3\sspeed)(1+\sspeed)}{1-\sspeed}}\biggl( \frac{|\wtt_*(x_\pm)|^{\frac{1+\sspeed}{1-\sspeed}}}{1-\sspeed}+ \Ord (t^\lambda)\biggr), \quad 0<t\leq t_\pm.
\end{equation}

Assuming now that $k\geq 3$ so that the derivatives $\del{x}\wtt$ and $\del{x}z$ converge pointwise uniformly on $\Tbb$ to
$\del{x}\wtt_*$ and $\del{x}z_*$ by the decay estimates  \eqref{asymp-C} and \eqref{asymp-D}, we see from
plugging in the asymptotic expansions \eqref{asymp-exp-B} and \eqref{asymp-exp-D} into \eqref{asymp-E} that
the asymptotic behaviour of the gradient density is determined by 
\begin{equation*}
\frac{\del{x}\rho(t,x_{\pm})}{\rho(t,x_\pm)}
= \frac{(1+\sspeed) \Delta_{\pm}}{ 1+ t^{\frac{2(1-3\sspeed)}{1-\sspeed} } \Bigl(|\wtt_*(x_\pm)|^\frac{2}{1-\sspeed}+ \Ord(t^\lambda)\Bigr)} 
\end{equation*}
where
\begin{align*}
\Delta_{\pm} =& \mp t^{\frac{2(1-3\sspeed)}{1-\sspeed}}\biggl( \frac{|\wtt_*(x_\pm)|^{\frac{1+\sspeed}{1-\sspeed}}}{1-\sspeed}+ \Ord (t^\lambda)\biggr)\Bigl(\del{x}\wtt_*(x_\pm)+\Ord\bigl(t^{(k-1)(1-3\sspeed)+1}\bigr)\Bigr)\\
&+ \Bigl(1+ t^{\frac{2(1-3\sspeed)}{1-\sspeed} } \Bigl(|\wtt_*(x_\pm)|^\frac{2}{1-\sspeed}+ \Ord(t^\lambda)\Bigr)\Bigr) \Bigl(\del{x}z_*(x_\pm)+\Ord\bigl(t^{k(1-3\sspeed)+1}\bigr)\Bigr).
\end{align*}
Simplifying, we arrive at the asymptotic formula
\begin{equation*}
\frac{\del{x}\rho(t,x_{\pm})}{\rho(t,x_\pm)} =(1+\sspeed)\Biggl(\mp\frac{\del{x}\wtt_*(x_\pm)}{(1-\sspeed)|\wtt_*(x_\pm)|} + \del{x}z_*(x_\pm)+\Ord\bigl(t^{\lambdat}\bigr)\Biggr),
\end{equation*}
which, after letting $t\searrow 0$, yields
\begin{equation*}
\lim_{t\searrow 0} \frac{\del{x}\rho(t,x_{\pm})}{\rho(t,x_\pm)} =(1+\sspeed)\biggl(\mp\frac{\del{x}\wtt_*(x_\pm)}{(1-\sspeed)|\wtt_*(x_\pm)|} + \del{x}z_*(x_\pm)\biggr).
\end{equation*} 

\end{proof}

\begin{rem} \label{rem-exist}
\begin{enumerate}[(i)]
$\;$

\item  If $\Wsc_*^{+}\neq \emptyset$ or $\Wsc^-_* \neq \emptyset$, then by Theorem \ref{thm-exist}.(c) the rescaled density   $t^{\frac{2(1+\sspeed)}{\sspeed-1}} \rho(t,x)$
will blow up somewhere at future timelike infinity, that is,
\begin{equation} \label{rho-vanish}
\sup_{x\in \Tbb} \limsup_{t\searrow 0}  t^{\frac{2(1+\sspeed)}{\sspeed-1}} \rho(t,x) = \infty,
\end{equation}
and at each spatial point $x_{\pm}\in \Wsc_*^{\pm}$, the rescaled four-velocity  $t^{\frac{2(1-2\sspeed)}{\sspeed-1}}\vt^i(t,x)$
will asymptote as $t\searrow 0$  to  the null-vector 
\begin{equation} \label{vel-asympt}
|w_*(x_\pm)|^{\frac{1}{1-\sspeed}} \bigl(-\delta^i_0 \pm \delta^i_1 \bigr).
\end{equation}

The asymptotic limit \eqref{vel-asympt} shows that the solutions from Theorem \ref{thm-exist} exhibit extreme tilt at spatial points at future timelike infinity where $\wtt_*$ does not vanish. It is worthwhile noting that \eqref{rho-vanish} is not consistent with the expected behavior of a fluid on an exponentially expanding cosmological spacetime. 
Indeed, this behaviour should be compared with that of small perturbations to the tilted homogeneous solutions of the relativistic Euler equations where it has been rigorously shown that the rescaled density $t^{\frac{2(1+\sspeed)}{\sspeed-1}} \rho$ converges to a \emph{bounded, strictly positive function} at timelike infinity
\cite{MarshallOliynyk:2022,Oliynyk:2021}. 
\item 
Letting
\begin{equation*}
\Wsc'_* = \{ x\in \Tbb \, | \, \del{x}\wtt_*(x)\neq 0\},
\end{equation*}
and assuming that $\del{}\Wsc_*^\pm \cap \Wsc'_* \neq 0$ and $k\geq 3$, there will, by the continuity of $\wtt_*$ and $\del{x}\wtt_*$, exist a sequences of points $x_{\pm}^j \in \Wsc_*^\pm$ such
that $\lim_{j\rightarrow \infty}\wtt_*(x_{\pm}^j)=0$
and $\lim_{j\rightarrow \infty}\del{x}\wtt_*(x_{\pm}^j) \neq 0$. In this case, the fractional density will, by Theorem \ref{thm-exist}.(c), blow up at future timelike infinity, that is,
\begin{equation*}
\sup_{x\in \Tbb} \limsup_{t\searrow 0} \frac{\del{x}\rho(t,x)}{\rho(t,x)} = \infty \AND \inf_{x\in \Tbb} \liminf_{t\searrow 0} \frac{\del{x}\rho(t,x)}{\rho(t,x)}=-\infty.
\end{equation*}
As discussed in the introduction, this behavior is of interest because it is not consistent with the expected behavior of the fractional density gradient on cosmological spacetimes. As noted in the introduction, this behavior was anticipated by Rendall in \cite{Rendall:2004}, and from the proof of Theorem \ref{thm-exist}, it is clear that it is due to anisotropies that are generated from the solution limiting as $t\searrow 0$ to different homogeneous solutions at spatial points $x_{\pm}$ where $\wtt_*(x_+)>0$ and $\wtt_*(x_-)<0$ for which there is no smooth transition between them as can be seen from the asymptotic limit \eqref{vel-asympt} of the rescaled fluid velocity.
\item The existence of open sets of initial data that yields solutions for which $\Wsc_*^{\pm}\neq \emptyset$ and $\del{}\Wsc_*^{\pm}\cap \Wsc_*'\neq \emptyset$ is established below in Theorem \ref{thm-non-empty}.
\item The initial value problem for the  \textit{asymptotic equation} associated to \eqref{T2-Eul-H.1} is given by
\begin{align}
\del{t}\Wtth &= \frac{3\sspeed-1}{t}\pip \Wtth \hspace{2.25cm} \text{in $(0,t_1]\times\Tbb$,}\label{T2-Eul-asymp.1}\\
\Wtth &= \Wtt(t_1)=(\ztt(t_1),\ztt(t_1))^{\tr} \hspace{0.5cm} \text{in $\{t_1\}\times \Tbb$,} 
\label{T2-Eul-asymp.2}
\end{align}
where we note that the asympotic equation
\eqref{T2-Eul-asymp.1} is obtained from \eqref{T2-Eul-H.1}
by dropping terms involving spatial derivatives. Due to the simplicity of the asymptotic equation, we can integrate it exactly to obtain the unique solution
\begin{equation} \label{asymp-IVP-sol}
\Wtth(t) = (\ztth(t),\wtth(t))^{\tr}=\biggl(\Bigl(\frac{t}{t_1}\Bigr)^{3\sigma-1}\ztt(t_1),\wtt(t_1)\biggr)^{\tr}.
\end{equation}
The estimates \eqref{wtt-asymp-est} and \eqref{z-asymp-est} from the proof of Theorem \ref{thm-exist} imply that the difference between the full solution $\Wtt(t)=(\ztt(t),\wtt(t))^{\tr}$ to the Fuchsian GIVP \eqref{T2-Eul-H.1}-\eqref{T2-Eul-H.2} from Theorem \ref{thm-exist} and the solution \eqref{asymp-IVP-sol} of the IVP for the asymptotic equation is bounded by
\begin{align*}
\norm{\wtt(t)-\wtth(t)}_{H^{k-1}} &\lesssim t_1^{(k-1)(1-3\sspeed)+1}- t^{(k-1)(1-3\sspeed)+1},
\\
\norm{\ztt(t)-\ztth(t)}_{H^{k-1}}&\lesssim t^{3\sigma-1}\bigl( t_1^{k(1-3\sspeed)+1}- t^{k(1-3\sspeed)+1} \bigr),
\end{align*}
for all $t\in (0,t_1]$. It is clear from these bounds and the assumption $k(1-3\sigma)+1>0$ that the error, as measured in the $H^{k-1}$ norm, between the full and asymptotic solutions can be made arbitrarily small, uniformly on the time interval $(0,t_1]$, by choosing $t_1$ sufficiently close to zero.  
This verifies that solutions from Theorem \ref{thm-exist} of the Fuchsian GIVP \eqref{T2-Eul-H.1}-\eqref{T2-Eul-H.2}  are \textit{asymptotically ODE dominated}, which was originally observed numerically in \cite{MarshallOliynyk:2022} using a different but equivalent Fuchsian formulation of the $\Tbb^2$-symmetric Euler equations; see also \cite{BMO:2023} where similar ODE dominated behavior was observed to occur numerically in solutions of the Einstein-Euler equations with Gowdy symmetry and a positive cosmological constant.
\item By Corollary  \ref{cor-Ac0Bc}, we know that the $\ell$-th spatial derivative of the matrix $\Btt(t,\wtt)$ blows up like $t^{\ell(1-3\sspeed)}$ near $t=0$. Because of this, the convergence rates in the decay estimates from Theorem \ref{thm-exist}.(a) can be improved by using a norm of lower regularity. Indeed, if $k-\ell-1>\frac{1}{2}$ for some $\ell \in \Zbb_{\geq 0}$, then we can replace the estimate \eqref{Btt-dWtt-est-A} in the proof of Theorem \ref{thm-exist} by
\begin{equation*}
\norm{\Btt(t,\wtt(t))\del{x}\Wtt(t) }_{H^{k-\ell-1}}\lesssim  \norm{\Btt(t,\wtt(t))}_{H^{k-\ell-1}}
\norm{\del{x}\Wtt(t)}_{H^{k-1}},
\end{equation*}
which yields the following replacement for \eqref{Btt-dWtt-est-B}:
\begin{equation*}
 \norm{\Btt(t,\wtt(t))}_{H^{k-\ell-1}} \lesssim t^{(k-\ell-1)(1-3 \sspeed)} (1+\norm{\wtt(t)}^{k-\ell-2}_{L^\infty})\bigl(1+\norm{\wtt(t)}_{H^{k-\ell-1}}\bigr).
\end{equation*}
Using this, it is then straightforward to modify the derivation of the  decay estimates from proof of Theorem \ref{thm-exist} to show that the estimates
\begin{align*}
\norm{\wtt(t)-\wtt_*}_{ C^{k-\ell-2,\frac{1}{2}}} \lesssim \norm{\wtt(t)-\wtt_*}_{H^{k-\ell-1}}  &\leq C\delta t^{(k-\ell-1)(1-3\sspeed)+1}
\intertext{and}
\norm{z(t)-z_*}_{C^{k-\ell-2,\frac{1}{2}}}\lesssim \norm{z(t)-z_*}_{H^{k-\ell-1}}  &\leq C\delta t^{(k-\ell)(1-3\sspeed)+1}
\end{align*}
hold for all $t\in (0,1]$. 
\end{enumerate}
\end{rem}

\section{Fractional density gradient blow-up at future timelike infinity\label{sec:frac-den-grad-blow-up}}
By Remark \ref{rem-exist}.(ii), we can verify that blow-up of the fractional density contrast at future timelike infinity actually occurs by showing that there exists initial data that yield asymptotic functions $\wtt_*$ for which the sets $\Wsc_*^{\pm}$ and $\del{}\Wsc_*^{\pm}\cap \Wsc_*'$ are non-empty. In this section, we establish the existence of such initial data; see Theorem \ref{thm-non-empty} for the precise statement. In particular, Theorem \ref{thm-non-empty}  confirms, in the restricted setting of relativistic fluids on exponentially expanding FLRW spacetimes, the conjecture by Rendall from \cite{Rendall:2004} regarding the blow-up of the fractional density gradient for sound speeds satisfying $1<\sspeed<1/3$.

In order to establish the existence of initial data sets that yields asymptotic functions $\wtt_*$ for which  $\Wsc_*^{\pm}$ and $\del{}\Wsc_*^{\pm}\cap \Wsc_*'$ are non-empty, we will solve a simpler system whose solutions we can accurately track over long time intervals\footnote{Recall that due to our compactification, the time interval $(0,1]$ corresponds to an infinitely long time interval in physical time.}, and at the same time, yield very accurate approximations to solutions of the full equations. 
The simpler system that we will solve is obtained by \textit{linearising} the full system \eqref{T2-Eul-H.1}. This results in the following initial value problem:  
\begin{align}
\del{t}\Wttt + \Bttt \del{x}\Wttt &= \frac{3\sspeed-1}{t}\pip \Wttt \hspace{0.5cm} \text{in $(0,1]\times\Tbb$,}\label{T2-Eul-linA.1}\\
\Wttt &=(\zttt_0,\wttt_0)^{\tr} \hspace{0.9cm} \text{in $\{1\}\times \Tbb$,}
\label{T2-Eul-linA.2}
\end{align}
where
\begin{align}
\Wttt&=(\zttt,\wttt)^{\tr}, \label{Wttt-def}\\
\Bttt &:= \Btt(t,0) = \begin{pmatrix}
0 & -1\\
-\sspeed & 0
\end{pmatrix}, \label{Bttt-def}
\end{align}
and in deriving the second equality in the above expression for $\Bttt$,
we used \eqref{cov}, \eqref{phi-def}, \eqref{Wtt-def} and \eqref{AcBc-def}.
We further note that this IVP for the  \textit{linearised equation} \eqref{T2-Eul-linA.1} can be expressed in second order form as
\begin{align}
\del{t}^2 \zttt-\sspeed \del{x}^2 \zttt
&= \frac{3\sspeed-1}{t}\del{t}\zttt-\frac{3\sspeed-1}{t^2}\zttt \hspace{0.7cm} \text{in $(0,1]\times \Tbb$}, \label{T2-Eul-linB.1}\\
\del{t}\wttt &= \sspeed \del{x}\zttt
 \hspace{3.45cm} \text{in $(0,1]\times \Tbb$}, 
 \label{T2-Eul-linB.2}\\
 (\zttt,\del{t}\zttt)&= \Bigl(\zttt_0,\del{x}\wttt_0+\frac{3\sspeed-1}{t}\zttt_0\Bigr)  \hspace{0.5cm}\text{in $\{1\}\times \Tbb$}, 
 \label{T2-Eul-linB.3}\\
 \wttt &= \wttt_0   \hspace{3.9cm}\text{in $\{1\}\times \Tbb$}.
 \label{T2-Eul-linB.4}
\end{align}

\subsection{Existence of solutions to the linearised IVP and error bounds}
We now turn to establishing the existence and uniqueness of solutions to the linearised IVP \eqref{T2-Eul-linA.1}-\eqref{T2-Eul-linA.2} and establishing an error estimate that measures the difference between solutions of the linearised IVP and the Fuchsian GIVP \eqref{T2-Eul-H.1}-\eqref{T2-Eul-H.2}.

\begin{prop} \label{prop-lin-approx}
Suppose $k\in \Zbb_{>\frac{3}{2}}$,  $\frac{1}{3}<\sspeed <\frac{k+1}{3k}$,  $0<\delta < \delta_0$ where $\delta_0>0$ is the constant from Theorem \ref{thm-exist}, $\Wttt(1)=(\zttt_0,\wttt_0)^{\tr}\in H^{k}(\Tbb,\Rbb^2)$, $\Wtt(1)=(\ztt_0,\wtt_0)^{\tr}\in H^{k}(\Tbb,\Rbb^2)$ satisfies
\begin{equation*}
\norm{\Wtt(1)}_{H^k} = \sqrt{\norm{\ztt_0}_{H^k}^2 + \norm{\wtt_0}_{H^k}^2} \leq \delta,
\end{equation*}
and 
\begin{equation*}
\Wtt=(\ztt, \wtt )^{\tr} \in C^0\bigl((0,1],H^{k}(\Tbb,\Rbb^2)\bigr)\cap C^1 \bigl((0,1],H^{k-1}(\Tbb,\Rbb^2)\bigr)
\end{equation*}
is the unique solution to the Fuchsian GIVP \eqref{T2-Eul-H.1}-\eqref{T2-Eul-H.2} from Theorem \ref{thm-exist}. Then
there exists a unique solution
\begin{equation*}
\Wttt=(\zttt, \wttt )^{\tr} \in C^0\bigl((0,1],H^{k}(\Tbb,\Rbb^2)\bigr)\cap C^1 \bigl((0,1],H^{k-1}(\Tbb,\Rbb^2)\bigr)
\end{equation*}
to the IVP \eqref{T2-Eul-linA.1}-\eqref{T2-Eul-linA.2}
for the linearised equation, and there exists a constant $C_0>0$ such that
\begin{equation*}
\norm{\Wttt(t)}_{C^{k-1},\frac{1}{2}}\lesssim \norm{\Wttt(t)}_{H^k} \leq C_0 \norm{\Wttt(1)}_{H^k}
\end{equation*}
for all $t\in (0,1]$. Moreover, $\Wttt(t)$ converges in  $\in H^{k-1}(\Tbb,\Rbb^2)$ as $t\searrow 0$ to a limit, denoted $\Wttt(0)$, and there exists a constant 
$C_1=C_1(\delta_0)$ such that
\begin{gather*}
\norm{\Wtt(t)-\Wttt(t)}_{C^{k-2},\frac{1}{2}}
\lesssim \norm{\Wtt(t)-\Wttt(t)}_{H^{k-1}} \leq C_1 \bigl(\norm{\Wtt(1)-\Wttt(1)}_{H^{k-1}}+\delta^2\bigr)
\end{gather*}
for all $t\in [0,1]$.
\end{prop}
\begin{proof}
The linearised equation \eqref{T2-Eul-linA.1} is symmetrisable as is easily verified by multiplying it on the left by the matrix
\begin{equation*}
\Attt^0 = \begin{pmatrix} \sspeed & 0\\
0 & 1 
\end{pmatrix}
\end{equation*}
to get
\begin{equation} \label{approx-1}
\Attt^0 \delt{t}\Wttt + \Attt^1 \del{x}\Wttt= \frac{3\sspeed-1}{t}\Attt^0 \pip\Wttt
\end{equation}
where 
\begin{equation*}
\Attt^1 = \begin{pmatrix} 0 & -\sspeed\\
-\sspeed & 0 
\end{pmatrix}.
\end{equation*}
Standard existence theory for linear symmetric hyperbolic equations, e.g. \cite[Thm.~2.6]{BenzoniSerre:2007}, then guarantees the existence of unique solution  
\begin{equation*}
\Wttt=(\zttt, \wttt )^{\tr} \in C^0\bigl((0,1],H^{k}(\Tbb,\Rbb^2)\bigr)\cap C^1 \bigl((0,1],H^{k-1}(\Tbb,\Rbb^2)\bigr)
\end{equation*}
to the IVP \eqref{T2-Eul-linA.1}-\eqref{T2-Eul-linA.2}. Moreover, differentiating \eqref{approx-1} with respect to $x$ yields
\begin{equation} \label{approx-2}
\Attt^0 \del{t}\del{x}^\ell \Wttt + \Attt^1 \del{x} \del{x}^\ell \Wttt= \frac{3\sspeed-1}{t}\Attt^0 \pip \del{x}^\ell\Wttt
\end{equation}
for all $\ell \in \Zbb$ satisfying $0\leq \ell \leq k$. Then, letting
\begin{equation*}
\nnorm{\Wttt}_0^2 := \ip{\Wttt}{\Attt^0\Wtt}
\end{equation*}
denote the $L^2$-energy norm and
\begin{equation*}
\nnorm{\Wttt}_s^2= \sum_{\ell=0}^s\nnorm{\del{x}^\ell \Wttt}_0^2
\end{equation*}
denote the higher order energy norms, which are clearly equivalent to the standard $H^s$-norms since $\sspeed>0$ by assumption, we obtain, by multiplying \eqref{approx-2} on the left by $\del{x}^\ell \Wttt^{\tr}$ and integrating by parts, the energy inequality
\begin{equation*}
-\frac{1}{2}\del{t}\nnorm{\del{x}^\ell \Wttt}_{0}^2= -\frac{3\sspeed-1}{t}\nnorm{ \pip \del{x}^\ell\Wttt}_0^2 \leq 0,
\end{equation*}
where in deriving this we employed the identities $(\pip)^2=\pip$, $(\pip)^{\tr}=\pip$, $(\Attt^0)^{\tr}=\Attt^0$,  $(\Attt^1)^{\tr}=\Attt^1$ and $[\pip,\Attt^0]=0$. Summing the above inequalities from $\ell=0$ to $\ell=k$
yields
\begin{equation*}
-\frac{1}{2}\del{t}\nnorm{\Wttt}_{k}^2 \leq 0
\end{equation*}
or equivalently
\begin{equation*}
-\del{t}\nnorm{\Wttt}_{k} \leq 0.
\end{equation*}
Integrating in time yields the uniform bound
\begin{equation*}
\nnorm{\Wttt(t)}_{k}  \leq \nnorm{\Wtt(1)}_{k}, \quad 0<t\leq 1,
\end{equation*} 
which, 
due to the norm equivalency 
$\nnorm{\cdot}_{k} \approx \norm{\cdot}_{H^k}$ and Sobolev's inequality (Theorem \ref{Sobolev}),
implies
\begin{equation} \label{approx-3}
\norm{\Wttt(t)}_{C^{k-1},\frac{1}{2}}\lesssim \norm{\Wttt(t)}_{k}  \lesssim \norm{\Wttt(1)}_{k}, \quad 0<t\leq 1.
\end{equation}

Next, we set
\begin{equation*}
\Vtt = \Wtt-\Wttt,
\end{equation*}
and observe from \eqref{T2-Eul-H.1} and
\eqref{T2-Eul-linA.1} that $\Vtt$ satisfies
\begin{equation} \label{approx-4}
\del{t}\Vtt + \Bttt\del{x}\Vtt = \frac{3\sspeed-1}{t}\pip \Vtt - F
\end{equation}
where
\begin{equation} \label{F-def}
F:= (\Btt(t,\wtt)-\Bttt)\del{x}\Wtt\oset{\eqref{Bttt-def}}{=}(\Btt(t,\wtt)-\Btt(t,0))\del{x}\Wtt.
\end{equation}
Multiplying \eqref{approx-4} on the left by $\Attt^0$ yields
\begin{equation} \label{approx-5}
\Attt^0\del{t}\Vtt + \Attt^1\del{x}\Vtt = \frac{3\sspeed-1}{t}\Attt^0 \pip \Vtt - \Attt^0 F.
\end{equation}
Applying $\del{x}^\ell$ to \eqref{approx-5} and then multiplying the resulting equation on the left by $\del{x}^\ell \Vtt^{\tr}$, we obtain, after integrating by parts, the energy inequality
\begin{equation*}
-\frac{1}{2}\del{t}\nnorm{\del{x}^\ell \Vtt}_{0}^2= -\frac{3\sspeed-1}{t}\nnorm{ \pip \del{x}^\ell\Vtt}_0^2 +\ip{\del{x}^\ell\Vtt}{\Attt^0\del{x}^\ell F}
\lesssim\norm{\del{x}^\ell\Vtt}_{L^2}\norm{\del{x}^\ell F}_{L^2},
\end{equation*}
where in deriving this we employed the Cauchy-Schwartz inequality and the identities $(\pip)^2=\pip$, $(\pip)^{\tr}=\pip$, $(\Attt^0)^{\tr}=\Attt^0$,  $(\Attt^1)^{\tr}=\Attt^1$ and $[\pip,\Attt^0]=0$. 
Summing the above inequalities from $\ell=0$ to $\ell=k-1$ and employing the norm equivalence $\nnorm{\cdot}_{k-1} \approx \norm{\cdot}_{H^{k-1}}$ yields the inequality
\begin{equation*}
-\frac{1}{2}\del{t}\nnorm{\Vtt}_{k-1}^2\lesssim \norm{F}_{H^{k-1}}\nnorm{\Vtt}_{k-1},
\end{equation*}
which after dividing through by $\nnorm{\Vtt}_{k-1}$ becomes
\begin{equation} \label{approx-6}
-\del{t}\nnorm{\Vtt}_{k-1}\lesssim \norm{F}_{H^{k-1}}.
\end{equation}
By Corollary \ref{cor-Ac0Bc}, the uniform bound $\norm{\Wtt(t)}_{C^{k-1,\frac{1}{2}}} \lesssim \norm{\Wtt(t)}_{H^k} \leq C(\delta_0)\delta$ from Theorem \ref{thm-exist}, and the calculus inequalities from Appendix \ref{calc}, in particular, Theorems \ref{Sobolev}, \ref{Product}.(ii) and \ref{Moser}, we observe that the source term \eqref{F-def} is bounded by 
\begin{align*}
\norm{F}_{H^{k-1}}\lesssim  \norm{\Btt(t,\wtt)-\Btt(t,0)}_{H^{k-1}}\norm{\del{x}\Wtt}_{H^{k-1}} 
& \leq t^{(k-1)(1-3\sspeed)}C(\delta_0)\delta^2.
\end{align*}
Substituting this into the energy inequality \eqref{approx-6} and integrating the resulting inequality in time, we find that
\begin{equation*}
\nnorm{V(t)}_{k-1} \leq C(\delta_0)\bigl(\nnorm{V(1)}_{k-1} + \delta^2\bigr), \quad 0<t\leq 1,
\end{equation*}
where in deriving this inequality we used the fact that
$(k-1)(1-3\sspeed)+1>0$. The estimate
\begin{equation} \label{Wttt-error}
\norm{\Wtt(t)-\Wttt(t)}_{C^{k-2},\frac{1}{2}}
\lesssim \norm{\Wtt(t)-\Wttt(t)}_{H^{k-1}} \leq C(\delta_0)\bigl( \norm{\Wtt(1)-\Wttt(1)}_{H^{k-1}}+ \delta^2\bigr), \quad 0<t\leq 1, 
\end{equation}
is then a direct consequence of the norm equivalence $\nnorm{\cdot}_{k-1} \approx \norm{\cdot}_{H^{k-1}}$ and Sobolev's inequality (Theorem \ref{Sobolev}).

To complete the proof, we need to show that the solution $\Wttt(t)$ to the linearised IVP converges in $H^{k-1}(\Tbb)$ as $t\searrow 0$. To establish convergence, we first apply the projection operator $\pi$ to \eqref{T2-Eul-linA.1} to obtain
\begin{equation*}
\del{t}\pi \Wttt +\pi \Bttt \del{x}\Wttt=0.
\end{equation*}
Integrating this is time gives
\begin{equation*}
\pi \Wttt(t_1)-\pi\Wttt(t_0)= -\int^{t_0}_{t_1} \pi \Bttt \del{x}\Wttt(s)\,ds,\quad 0<t_1 < t_0\leq 1.
\end{equation*}
Applying the $H^{k-1}$ norm to this expression, we find, with the help of the bound \eqref{approx-3} and the triangle inequality, that
\begin{equation*} 
\norm{\pi \Wttt(t_1)-\pi\Wttt(t_0)}_{H^{k-1}}\leq \int^{t_0}_{t_1} \norm{\pi \Bttt \del{x}\Wttt(s)}_{H^{k-1}}\,ds \lesssim \norm{\Wttt(1)}(t_0-t_1),\quad  0<t_1 < t_0\leq 1.
\end{equation*}
From this estimate, it then follows that $\pi \Wttt(t)$ converges in $H^{k-1}(\Tbb)$ as $t\searrow 0$.

Next, applying the projection operator $\pip$ to \eqref{T2-Eul-linA.1}, we have
\begin{equation*}
\del{t}\pip (t^{1-3\sspeed}\Wttt) + \pip \Bttt\del{x} (t^{1-3\sspeed}\Wttt) = 0. 
\end{equation*}
Integrating this in time and arguing as above yields
\begin{equation*} 
\norm{\pip t_1^{1-3\sspeed}\Wttt(t_1)-\pip t^{1-3\sspeed}_0\Wttt(t_0)}_{H^{k-1}}\lesssim \norm{\Wttt(1)}(t_0^{2-3\sspeed}-t_1^{2-3\sspeed}),\quad  0<t_1 < t_0\leq 1.
\end{equation*}
Since $\frac{1}{3}<\sspeed < \frac{k+1}{3k}$ and $k\geq 2$, we have that $2-3\sspeed>0$, and hence, it follows from the above estimate that $\pip t^{1-3\sspeed}\Wttt(t)$ converges in $H^{k-1}(\Tbb)$ as $t\searrow 0$. From this, we deduce that $ \pip\Wttt(t)$ converges to $0$ in $H^{k-1}(\Tbb)$ as $t\searrow 0$. Now, because $\text{id}_{\Rbb^2} = \pi + \pip$, it follows immediately that $\Wttt(t)$ converges to a limit, which we denote by $\Wttt(0)$, in $H^{k-1}(\Tbb)$ as $t\searrow 0$. Due to this convergence, we note that the estimate \eqref{Wttt-error} must also hold for $t=0$. This completes the proof.
\end{proof}

\subsubsection{The general solution to the linear wave equation} Proposition \ref{prop-lin-approx} shows that the solutions $\Wtt$ of the Fuchsian GIVP \eqref{T2-Eul-H.1}-\eqref{T2-Eul-H.2} can be accurately approximated by a solution $\Wttt$ of the linearised IVP \eqref{T2-Eul-linA.1}-\eqref{T2-Eul-linA.2}, where by accurate we mean the remainder in the error estimate from Proposition \ref{prop-lin-approx} is of size $\delta^2$ compared to the size of the initial data for the Fuchsian GIVP, which is of order $\delta$. However, for this error estimate to be useful for analysing the asymptotic function $\wtt_*$, c.f.~Theorem \ref{thm-exist}.(a), we need to be able to determine the precise behaviour of the linearised solution $\Wttt(t)$ over long time intervals, in particular, all the way to future timelike infinity at $t=0$. 
Here, we will accomplish this by using separation of variables to find the general solution of the linear wave equation \eqref{T2-Eul-linB.1} from which we can obtain the general solution of the linearised equation \eqref{T2-Eul-linA.1}. From the general solution, we then select particular solutions with properties that we need in order to establish the existence of open sets of initial data for the Fushsian GIVP that generate solutions that  yield asymptotic functions $\wtt_*$ for which $\Wsc_*^{\pm}$ and $\del{}\Wsc_*^{\pm}\cap \Wsc_*'$ are non-empty.

Without loss of generality, we will assume that the period of the torus $\Tbb$ is $1$ so that we can interpret functions on $\Tbb$ as functions on $[0,1]$ that satisfy periodic boundary conditions. Then expanding a solution $\zttt(t,x)$ of the linear wave equation \eqref{T2-Eul-linB.1} as a Fourier series in the variable $x$, we find after a short calculation that $\zttt$ can be represented as 
\begin{equation} \label{sep-1}
\zttt(t,x) = T_0(t) + \sum_{n=1}^\infty T_n(t)\bigl(a_n \cos(2\pi n x)+b_n\sin(2\pi n x)\bigr)
\end{equation}
where the $T_n(t)$ satisfy
\begin{equation} \label{sep-2}
t^2 T_n''(t) - (3\sspeed -1) t T_n'(t) + \bigl(3\sspeed -1 + 4\sspeed \pi^2 n^2 t^2\bigr)T_n(t) = 0
\end{equation}
and $a_n$, $b_n\in \Rbb$ are arbitrary constants. 
The general solution of the ODE \eqref{sep-2} can be represented in terms of Bessel functions of the first and second kind, $J_\alpha (t)$ and $Y_\alpha(t)$, respectively, as
\begin{equation} \label{sep-3}
T_n(t)=t^{\frac{3\sspeed}{2}}\Bigl(c_n J_{\frac{3\sspeed-2}{2}}\bigl(2\sqrt{\sspeed} \pi n t\bigr)+d_n Y_{\frac{3\sspeed-2}{2}}\bigl(2\sqrt{\sspeed} \pi n t\bigr)\Bigr)    
\end{equation}
where $c_n,d_n\in \Rbb$ are arbitrary constants. Substituting this into \eqref{sep-1}  yields the general solution of the linear wave equation \eqref{T2-Eul-linA.1}.

\subsubsection{Particular solutions of the linearised equations}

By \eqref{sep-1} and \eqref{sep-3}, we know that 
\begin{equation} \label{part-sol-1}
\zttt{}(t,x) = T_n(t)\cos(2\pi n x), \quad n\in \Zbb_{>0},
\end{equation}
are solutions of the linear wave equation \eqref{T2-Eul-linB.1}.  By \eqref{T2-Eul-linB.3}, this solution leads to the following initial data for the linearised equation \eqref{T2-Eul-linA.2}:
\begin{equation} \label{part-sol-2}
\Wttt(1,x)=(\zttt_0(x),\wttt_0(x))^{\tr} = \biggl( T_{n}(1)\cos(2\pi n x),\frac{1}{2\pi n} \bigl(T_{n}'(1)-(3\sspeed-1)T_{n}(1)\bigr)\sin(2\pi n  x) +b \biggr)^{\tr}
\end{equation}
where $b\in \Rbb$ is an arbitrary constant. Then setting $b=0$, it follows from \eqref{T2-Eul-linB.2}, \eqref{part-sol-1} and \eqref{part-sol-2} that  
\begin{equation} \label{part-sol-4}
\Wttt(t,x) = \begin{pmatrix} \zttt(t,x)
\\
\wttt(t,x)
\end{pmatrix}= 
 \begin{pmatrix}
 T_{n}(t)\cos(2\pi n x) \\
S_{n}(t)\sin(2\pi n x)
\end{pmatrix},  \quad n\in \Zbb_{>0},
\end{equation}
where
\begin{equation*}
S_n(t)= \frac{1}{2\pi n}\bigl(T_n'(1)-(3\sspeed-1)T_n(1)\bigr) + 2\pi n \sspeed \int^1_t T_n(\tau)\, d\tau,
\end{equation*}
is the unique solution to the linearised equation \eqref{T2-Eul-linA.2} that is generated by the initial data \eqref{part-sol-2} with $b=0$.
Using the Bessel recurrence relation from \cite[Eqns.(6.5)-(6.6)]{Korenev:2002}, it is not difficult to verify that the functions $S_n(t)$ can be expressed as
\begin{align*}
S_n(t) =& -t^{\frac{ 3\sspeed}{2}} \sqrt{\sspeed} \Bigl( c_n J_{\frac{3 \sspeed}{2}}\bigl(2 \sqrt{\sspeed} n \pi 
   t\bigr)+ d_n Y_{\frac{3
   \sspeed}{2}}\bigl(2 \sqrt{\sspeed} n \pi  
   t\bigr)\Bigr). 
\end{align*}
Then with the help of the power series representations of the Bessel functions, see equations (1.7) and (2.1) from \cite{Korenev:2002}, we can compute the limits of $S_n(t)$ and $T_n(t)$ as $t\searrow 0$ to get
\begin{align}
S_n(0)=& d_n \pi ^{-\frac{3 \sigma }{2}-1} \sigma ^{\frac{1}{2}-\frac{3
   \sigma }{4}} n^{-\frac{3 \sigma}{2}} \Gamma \Bigl(\frac{3 \sigma
   }{2}\Bigr), \label{Sn-lim}
\end{align}
and
\begin{equation} \label {Tn-lim}
T_n(0)  = 0, 
\end{equation}
respectively.

\subsubsection{Special initial data\label{Ckn-class-idata}}
For use below, we first define, for $k,n\in \Zbb_{>0}$, the following smooth pairs of functions: 
\begin{equation} \label{ztth-wtth-def}
(\ztth_0^{n,k}(x),\wtth_0^{n,k}(x)) = q_{n,k}\Bigl( T_{n}(1)\cos(2\pi n x),\frac{1}{2\pi n} \bigl(T_{n}'(1)-(3\sspeed-1)T_{n}(1)\bigr)\sin(2\pi n  x)\Bigr)
\end{equation}
where the constants $c_n,d_n\in \Rbb$ in $T_n(t)$
can be freely chosen as long as
\begin{equation} \label{dn-fix}
d_n \neq 0
\end{equation}
and
the constants $q_{k,n}$ are chosen to ensure that
\begin{equation}\label{ckn-def}
\sqrt{\norm{\ztth_0}_{H^k}^2+ \norm{\wtth_0}_{H^k}^2}=1.
\end{equation}
Then for any $\delta>0$, if we let $\Wttt(t,x)=(\zttt(t,x),\wttt(t,x))^{\tr}$
denote the solution of the linearised IVP  \eqref{T2-Eul-linA.1}-\eqref{T2-Eul-linA.2} generated from the initial data
$\Wttt(1) = \frac{\delta}{2} (\ztth_0^{n,k},\wtth_0^{n,k})^{\tr}$, we have by \eqref{part-sol-4},  \eqref{Sn-lim}-\eqref{Tn-lim}, \eqref{ztth-wtth-def}, and \eqref{dn-fix} 
that
\begin{equation} \label{wttt-lim}
\Wttt(0,x) =\biggl(0, \frac{\delta}{2}q_{k,n} S_n(0)\sin(2\pi n x)\biggr)^{\tr}, \quad k,n\in \Zbb_{>0},
\end{equation}
where
\begin{equation} \label{Sn-lbnd}
q_{k,n}\neq 0 \AND S_n(0) \neq 0.
\end{equation}

\subsection{Non-emptiness of the sets $\Wsc_*^\pm$ and  $\del{}\Wsc_*^\pm\cap \Wsc'_*$}
For a given asymptotic function $\wtt_*$, the sets $\Wsc_*^\pm$ were defined above in Theorem \ref{thm-exist}.(b) and the set $\Wsc'_*$ was defined in Remark \ref{rem-exist}.(ii). In the following theorem, we establish the existence of open sets of initial data that generate asymptotic functions $\wtt_*$ for which the sets $\Wsc_*^\pm$ and $\del{}\Wsc_*^\pm \cap \Wsc'_*$ are non-empty, which leads, among other things, to vanishing of the rescaled density and the blow-up of the fractional density gradient at future timelike infinity. More precisely, for given $n\in \Zbb_{>0}$,
$k\in \Zbb_{>\frac{3}{2}}$ and $\delta \in (0,\min\{\delta_0,1\})$ where $\delta_0>0$ is the constant from Theorem \ref{thm-exist}, we will consider the open set of initial data in $H^{k}(\Tbb,\Rbb^2)$
that consists of $\Wtt(1)=(\ztt_0,\wtt_0)^{\tr}$
satisfying
\begin{equation} \label{open-idata-A}
\norm{\Wtt(1)-2^{-1}\delta(\ztth^{k,n}_0,\wtth^{k,n}_0)^{\tr}}_{H^k} <\frac{\delta^2}{2}.
\end{equation}
By the triangle inequality 
and \eqref{ckn-def}, we observe that 
\begin{equation}\label{open-idata-B}
\norm{\Wtt(1)}_{H^k} \leq \norm{\Wtt(1)-2^{-1}\delta(\ztth^{k,n}_0,\wtth^{k,n}_0)^{\tr}}_{H^k} + \frac{1}{2}\delta\norm{(\ztth^{k,n}_0,\wtth^{k,n}_0)^{\tr}}_{H^k} < \frac{\delta(1+\delta)}{2} \leq \delta<\delta_0.
\end{equation}
The importance of this bound is that it implies that $\Wtt(1)$ satisfies the assumptions from Theorem \ref{thm-exist}, and consequently, provided $\frac{1}{3}<\sspeed <\frac{k+1}{3k}$, this initial data will, by Theorem \ref{thm-exist}, generate a unique solution of the Fuchsian GIVP \eqref{T2-Eul-H.1}-\eqref{T2-Eul-H.2}.

\begin{thm}\label{thm-non-empty}
Suppose $n\in \Zbb_{>0}$, $k\in \Zbb_{>\frac{3}{2}}$,  $\frac{1}{3}<\sspeed <\frac{k+1}{3k}$,  $\delta \in (0,\min\{\delta_0,1\})$ where $\delta_0>0$ is the constant from Theorem \ref{thm-exist},  
$\Wtt(1)=(\ztt_0,\wtt_0)^{\tr}\in H^{k}(\Tbb,\Rbb^2)$ satisfies \eqref{open-idata-A},
\begin{equation*}
\Wtt=(\ztt, \wtt )^{\tr} \in C^0\bigl((0,1],H^{k}(\Tbb,\Rbb^2)\bigr)\cap C^1 \bigl((0,1],H^{k-1}(\Tbb,\Rbb^2)\bigr)
\end{equation*}
is the unique solution to the Fuchsian GIVP \eqref{T2-Eul-H.1}-\eqref{T2-Eul-H.2} from Theorem \ref{thm-exist} generated by the initial data $\Wtt(1)$, $\wtt_*\in H^{k-1}(\Tbb)$ is the limit function from Theorem \ref{thm-exist}.(a), and $\Wsc_*^\pm\subset \Tbb$ are the sets defined in Theorem \ref{thm-exist}.(c). Then there exists a $\delta_1\in (0,\min\{\delta_0,1\})$ such that if $0<\delta<\delta_1$, then
the sets $\Wsc_*^\pm$ are non-empty, and 
\begin{equation*}
\sup_{x\in \Tbb} \limsup_{t\searrow 0} t^{\frac{2(1+\sspeed)}{\sspeed-1}} \rho(t,x) = \infty.
\end{equation*}
Moreover, if $k\in \Zbb_{>\frac{5}{2}}$ and $\Wsc_*'$ is the set defined in Remark \ref{rem-exist}.(ii), then
$\del{}\Wsc_*^\pm \cap \Wsc'_*\neq \emptyset$
and 
\begin{equation*}
\sup_{x\in \Tbb} \limsup_{t\searrow 0} \frac{\del{x}\rho(t,x)}{\rho(t,x)} = \infty \AND \inf_{x\in \Tbb} \liminf_{t\searrow 0} \frac{\del{x}\rho(t,x)}{\rho(t,x)}=-\infty.
\end{equation*}  
\end{thm}

\begin{proof}
Fixing  $n\in \Zbb_{>0}$, $k\in \Zbb_{>\frac{3}{2}}$,  $\frac{1}{3}<\sspeed <\frac{k+1}{3k}$, and  $\delta \in (0,\min\{\delta_0,1\})$ where $\delta_0>0$ is the constant from Theorem \ref{thm-exist}, we
let 
\begin{equation*}
\Wtt=(\ztt, \wtt )^{\tr} \in C^0\bigl((0,1],H^{k}(\Tbb,\Rbb^2)\bigr)\cap C^1 \bigl((0,1],H^{k-1}(\Tbb,\Rbb^2)\bigr)
\end{equation*}
denote the unique solution to the Fuchsian GIVP \eqref{T2-Eul-H.1}-\eqref{T2-Eul-H.2} from Theorem \ref{thm-exist} generated by the initial data $\Wtt(1)=(\ztt_0,\wtt_0)^{\tr}\in H^{k}(\Tbb,\Rbb^2)$ satisfying \eqref{open-idata-A}, and hence also \eqref{open-idata-B}. We also let
\begin{equation*}
\Wttt=(\ztt, \wtt )^{\tr} \in C^0\bigl((0,1],H^{k}(\Tbb,\Rbb^2)\bigr)\cap C^1 \bigl((0,1],H^{k-1}(\Tbb,\Rbb^2)\bigr)
\end{equation*}
denote the unique solution to the IVP \eqref{T2-Eul-linA.1}-\eqref{T2-Eul-linA.2} for the linearised equation, see Proposition \ref{prop-lin-approx}, generated by the initial data $\Wttt(1) = 2^{-1}\delta(\ztth^{k,n}_0,\wtth^{k,n}_0)^{\tr}$.
Then
it follows from Proposition \ref{prop-lin-approx}
and \eqref{open-idata-A} that
\begin{equation*}
\norm{\delta^{-1}(\Wtt(0)-\Wttt(0))}_{C^{k-2},\frac{1}{2}}
 \leq C\delta
\end{equation*}
where the constant $C>0$ is independent of
$\delta \in (0,\min\{\delta_0,1\})$. But $\Wtt(0)=(0,\wtt_*)$ by Theorem \ref{thm-exist}.(a), and $\Wttt(0)=(0,\delta \wtth_*^{k,n})$, 
where
\begin{equation} \label{wtth-def}
\wtth^{k,n}_*(x)= 2^{-1}q_{k,n}S_n(0)\sin(2\pi n x)
\end{equation} by \eqref{wttt-lim}, and $q_{k,n}, S_n(0)\neq 0$ by \eqref{Sn-lbnd}. Thus, we have
\begin{equation} \label{wtt*-wtth*}
\norm{\delta^{-1}\wtt_*-\wtth_*^{k,n}}_{C^{k-2},\frac{1}{2}}
 \leq C\delta
\end{equation}
for all $\delta \in (0,\min\{\delta_0,1\})$. Since $k\geq 2$ by assumption, we deduce from \eqref{wtt*-wtth*} that $\delta^{-1}\wtt_*$ converges uniformly on $\Tbb$ to $\wtth_*^{k,n}$ as $\delta \searrow 0$. Noting from \eqref{wtth-def} that the graph of $\wtth^{k,n}_*$ crosses $0$ more than once, it follows from the uniform convergence that graph of $\delta^{-1}\wtt_*$ must do the same for $\delta$ chosen sufficiently small. In particular, we conclude the existence of a $\delta_1\in (0,\min\{\delta_0,1\})$ such that $\Wsc^{\pm}_*\neq \emptyset$ provided $\delta \in (0,\delta_1)$. Assuming $\delta$ is chosen so that this holds, then by Remark \ref{rem-exist}.(i), the rescaled density $t^{\frac{2(1+\sspeed)}{\sspeed-1}} \rho$ must blow up at future timelike infinity in the sense that
\begin{equation*}
\sup_{x\in \Tbb} \limsup_{t\searrow 0} t^{\frac{2(1+\sspeed)}{\sspeed-1}} \rho(t,x) = \infty.
\end{equation*}

If we now assume that $k\geq 3$, then \eqref{wtt*-wtth*} shows that $\delta^{-1}\wtt_*$ and $\delta^{-1}\del{x}\wtt_*$ converges uniformly on $\Tbb$ to $\wtth_*^{k,n}$ and $\del{x}\wtth_*^{k,n}$, respectively, as $\delta \searrow 0$. From the formula \eqref{wtth-def}, it is clear that there exist points $x_*\in \Tbb$ such that $\wtth_*^{k,n}(x_*)=0$ and $\del{x}\wtth_*^{k,n}(x_*)\neq 0$. Due to the uniform convergence of $\delta^{-1}\wtt_*$ and $\delta^{-1}\del{x}\wtt_*$ to $\wtth_*^{k,n}$ and $\del{x}\wtth_*^{k,n}$, respectively, this property will also hold for  $\delta^{-1}\wtt_*$ provided $\delta$ is chosen sufficiently small. Thus, by shrinking $\delta_1$ if necessary, there will exist at least one point $y_*\in \Tbb$ such that $\wtt_*(y_*)=0$ and $\del{x}\wtt_*(y_*)\neq 0$ provided $\delta$ is chosen so that $0<\delta<\delta_1$. Assuming this is the case, then we must have that $\del{}\Wsc^{\pm}_*\cap \Wsc_*'\neq \emptyset$ and it follows from Remark \ref{rem-exist}.(ii) that the fractional gradient density $\frac{\del{x}\rho}{\rho}$ will blow-up at future timelike infinity in the sense that
\begin{equation*}
\sup_{x\in \Tbb} \limsup_{t\searrow 0} \frac{\del{x}\rho(t,x)}{\rho(t,x)} = \infty \AND \inf_{x\in \Tbb} \liminf_{t\searrow 0} \frac{\del{x}\rho(t,x)}{\rho(t,x)}=-\infty.
\end{equation*}  
\end{proof}

\section{Beyond leading order asymptotics\label{sec:beyond}}

The leading order asymptotics near $t=0$ of the $\Tbb^2$-symmetric solutions of the relativistic Euler equations from Theorem \ref{thm-exist} are rigorously determined by the decay estimates stated there for the variables $(z,\wtt)$. In the proof of Theorem \ref{thm-exist}, we used this asymptotic information to derive asymptotic expansions near $t=0$ for the rescaled fluid density, fluid four velocity and the fractional density gradient, and calculate their limits as $t\searrow 0$ at spatial points $x\in \Tbb$ where $\wtt_*(x)\neq 0$. In order to obtain more information about the solutions near $t=0$, we need to derive beyond leading order asymptotic expansions, or in other words, improved decay estimates. 
As a first step towards this goal, we derive a \textit{refined approximate solution} in this section that yields an improvement to the decay estimates established in Theorem \ref{thm-exist}.(a) and Remark \ref{rem-exist}.(v). In principle, this process could be iterated to construct a sequence of approximations that would yield corresponding improvements to the decay estimates. However, we will not consider these higher order approximations here.  

As an application of the refined approximation, we use it to derive asymptotic expansions for the rescaled fluid density, fluid four velocity, and the fractional density gradient near $t=0$ at spatial points $x\in \Tbb$ where $\wtt_*(x)= 0$, which completes the results established in Theorem \ref{thm-exist}.(c). 
However, before we consider the refined approximation, we first establish an estimate in the following lemma that will be needed below to derive improved error estimates. 
\begin{lem} \label{Btt-diff-lem}
Suppose $k\in \Zbb_{>\frac{1}{2}}$ and $\sigma \in [0,1)$. Then there exists a constant $C>0$ such that
\begin{equation*}
\norm{\Btt(t,\wtt)-\Btt(t,\vtt)}_{H^{k}} \leq C t^{(k+1)(1-3\sspeed)}\bigl(1+\norm{\wtt}_{H^k}^{k}+\norm{\vtt}_{H^k}^{k}\bigr)\norm{\wtt-\vtt}_{H^k}
\end{equation*}
for all $\wtt,\vtt \in H^k(\Tbb)$ and $t\in (0,1]$.
\end{lem}
\begin{proof} 
Recalling that the matrix $\Btt(t,\wtt)$ is defined by \eqref{AcBc-def}, we can, for $\xi,\eta\in \Rbb$, use the Fundamental Theorem of Calculus to express the difference
$\Btt(t,\xi)-\Btt(t,\eta)$ as
\begin{equation*}
\Btt(t,\xi)-\Btt(t,\eta) = \int_0^1 \frac{d}{ds} \Btt(t,\eta +s(\xi-\eta))\, ds = \int_0^1\del{\wtt}\Btt(t,\eta+s(\xi-\eta))\, ds(\xi-\eta).
\end{equation*}
Setting  
\begin{equation}\label{Rtt-def}
\Rtt(t,\xi,\eta) = \int_0^1\del{\wtt}\Btt(t,\eta+s(\xi-\eta))\, ds,
\end{equation}
we have that
\begin{equation} \label{Btt-diff-rep}
\Btt(t,\xi)-\Btt(t,\eta)=\Rtt(t,\xi,\eta)(\xi-\eta).
\end{equation}
Differentiating \eqref{Rtt-def} with respect to $\xi$ and $\eta$ yields
\begin{equation*}
\del{\xi}^\ell \del{\eta}^{m}\Rtt(t,\xi,\eta) = \int_0^1\del{\wtt}^{\ell+m+1}\Btt(t,\eta+s(\xi-\eta))s^\ell (1-s)^m\, ds.
\end{equation*}
From this formula and Corollary \ref{cor-Ac0Bc}, we deduce the existence of a constant $C=C(\ell,m)>0$ such that
\begin{equation} \label{Rtt-bnds}
|\del{\xi}^\ell \del{\eta}^m\Rtt(t,\xi,\eta)|_{\op} \leq C t^{(\ell+m+1)(1-3\sspeed)}   
\end{equation}
for all $(t,\xi,\eta)\in (0,1]\times \Rbb \times \Rbb$. 
Now, suppose $\vtt,\wtt \in H^{k}(\Tbb)$. Then we can estimate the difference $\Btt(t,\wtt)-\Btt(t,\vtt)$ as follows:
\begin{align*}
\norm{\Btt(t,\wtt)-\Btt(t,\vtt)}_{H^k} \lesssim&  
\norm{\Rtt(t,\wtt,\vtt)}_{H^k}\norm{\wtt-\vtt}_{H^k} \\
\lesssim & t^{(k+1)(1-3\sspeed)}\bigl(1+\norm{\wtt}_{L^\infty}^{k-1}+\norm{\vtt}_{L^\infty}^{k-1}\bigr)(\norm{\wtt}_{H^k}+\norm{\vtt}_{H^k})\norm{\wtt-\vtt}_{H^k}\\
\lesssim & t^{(k+1)(1-3\sspeed)}\bigl(1+\norm{\wtt}_{H^k}^{k}+\norm{\vtt}_{H^k}^{k}\bigr)\norm{\wtt-\vtt}_{H^k},
\end{align*}
where in deriving the first, second and third inequalities we employed \eqref{Btt-diff-rep} and Theorem \ref{Product}.(ii), \eqref{Rtt-bnds} and Theorem \ref{Moser}, and Theorem \ref{Sobolev}, respectively. This completes the proof.
\end{proof}

For the remainder of this section, we will assume that $k \in \Zbb_{>\frac{5}{2}+\ell}$ for some $\ell\in \Zbb_{\geq 0}$, $\frac{1}{3}<\sspeed < \frac{k+1}{3k}$, and 
\begin{equation*}
\Wtt=(\ztt, \wtt )^{\tr} \in C^0\bigl((0,1],H^{k}(\Tbb,\Rbb^2)\bigr)\cap C^1 \bigl((0,1],H^{k-1}(\Tbb,\Rbb^2)\bigr)
\end{equation*}
is a solution of the Fuchsian GIVP \eqref{T2-Eul-H.1}-\eqref{T2-Eul-H.2} from Theorem \ref{thm-exist}.
Then, letting $\Btt_{\alpha\beta}(t,\wtt)$, $1\leq \alpha,\beta \leq 2$, denote the component of the matrix $\Btt(t,\wtt)$, see \eqref{AcBc-def}, and 
\begin{equation*}
z = t^{1-3\sspeed}\ztt
\end{equation*}
be as defined previously in Theorem \ref{thm-exist}.(a), we can express the Fuchsian formulation \eqref{T2-Eul-H.1} of the $\Tbb^2$-symmetric relativistic Euler equations, see \eqref{pi-Eul} and \eqref{pip-Eul}, as
\begin{align}
\del{t}z &= F(t,z,\wtt) :=-\Btt_{11}(t,\wtt)\del{x}z- t^{1-3\sspeed}\Btt_{12}(t,\wtt) \del{x}\wtt , \label{z-ev}\\
\del{t}\wtt & = G(t,z,\wtt) := -t^{3\sspeed-1}\Btt_{21}(t,\wtt)\del{x}z- \Btt_{22}(t,\wtt) \del{x}\wtt. \label{wtt-ev}
\end{align}
We now define a \textit{refined approximate solution}, denoted $(\zb,\wttb)$, by demanding that it solve the initial value problem 
\begin{align}
\del{t}\zb &= F(t,z_*,\wtt_*)\hspace{0.5cm}\text{in $[0,1]\times \Tbb$,} \label{zb-ev}\\
\del{t}\wttb &= G(t,z_*,\wtt_*)\hspace{0.5cm}\text{in $[0,1]\times \Tbb$,} \label{wttb-ev}\\
(\zb,\wttb) &= (z_*,\wtt_*)\hspace{1.15cm}\text{in $\{0\}\times \Tbb$.} \label{zb-wttb-idata}
\end{align}
Since $\wtt_*, z_*\in H^{k-1}(\Tbb)$, $k-\ell-2\geq 1$ and $1-3\sspeed <0$, we can bound $F(t,z_*,\wtt_*)$ in $H^{k-\ell-2}(\Tbb)$ as follows:
\begin{align}
\norm{F(t,z_*,\wtt_*)}_{H^{k-\ell-2}} \lesssim& \norm{\Btt_{11}(t,\wtt_*)}_{H^{k-\ell-2}} \norm{\del{x}z_*}_{H^{k-\ell-2}}+t^{1-3\sspeed}\norm{\Btt_{12}(t,\wtt_*)}_{H^{k-\ell-2}} \norm{\del{x}\wtt_*}_{H^{k-\ell-2}}, \notag \\
\lesssim& \bigl(1+\norm{\wtt_*}_{L^\infty}^{k-\ell-3}\bigr)\bigl( t^{(k-\ell-2)(1-3\sspeed)}\norm{z_*}_{H^{k-\ell-1}} +  t^{(k-\ell-1)(1-3\sspeed)}\norm{\wtt_*}_{H^{k-\ell-1}}\bigr)\notag \\
\lesssim & \bigl(1+\norm{\wtt_*}_{H^{k-\ell-1}}^{k-\ell-3}\bigr)\bigl(\norm{z_*}_{H^{k-\ell-1}} + \norm{\wtt_*}_{H^{k-\ell-1}}\bigr) t^{(k-\ell-1)(1-3\sspeed)}, \label{F*-bnd}
\end{align}
where in deriving the first, second and third inequalities we used Theorem \ref{Product}.(ii), Corollary \ref{cor-Ac0Bc} and Theorem \ref{Moser}, and Theorem \ref{Sobolev}, respectively. Further, we can bound $G(t,z_*,\wtt_*)$ similarly by
\begin{equation} 
\norm{G(t,z_*,\wtt_*)}_{H^{k-\ell-2}} \lesssim
\bigl(1+\norm{\wtt_*}_{H^{k-\ell-1}}^{k-\ell-3}\bigr)\bigl(\norm{z_*}_{H^{k-\ell-1}} + \norm{\wtt_*}_{H^{k-\ell-1}}\bigr) t^{(k-\ell-2)(1-3\sspeed)}. \label{G*-bnd}
\end{equation}
Due to these estimates and the fact that $(k-\ell-1)(1-3\sspeed)>-1$, the IVP \eqref{zb-ev}-\eqref{zb-wttb-idata} has
a unique solution 
\begin{equation*}
(\zb,\wttb) \in C^0\bigl([0,T],H^{k-\ell-2}(\Tbb))\cap C^1((0,1],H^{k-\ell-2}(\Tbb)\bigr)
\end{equation*}
that can be expressed as 
\begin{align}
\zb(t,x) &= z_*(x) + \int_0^t F(s,z_*(x),\wtt_*(x))\, ds, \label{zb-sol} \\
\wttb(t,x) &= \wtt_*(x) + \int_0^t G(s,z_*(x),\wtt_*(x))\, ds,  \label{wttb-sol}
\end{align}
for all  $(t,x)\in [0,1]\times \Tbb$. Moreover, 
the decay estimates
\begin{align}
\norm{\zb(t)-\zb_*}_{H^{k-2}} \leq& C t^{(k-\ell-1)(1-3\sspeed)+1}, \label{zb-sol-decay}\\
\norm{\wttb(t)-\wttb_*}_{H^{k-2}} \leq& C t^{(k-\ell-2)(1-3\sspeed)+1}, \label{wttb-sol-decay}
\end{align}
hold for all $t\in (0,1]$, where $C$ is a constant of the form $C=C\bigl(\norm{z_*}_{H^{k-\ell-1}},\norm{\wtt_*}_{H^{k-\ell-1}}\bigr)$.

Next, we observe that the differences $z-\zb$ and $\wtt-\wttb$  satisfy
\begin{align}
\del{t}(z-\zb) &= F(t,z,\wtt)-F(t,z_*,\wtt_*), \label{diff-z-ev} \\ 
\del{t}(\wtt-\wb) &= F(t,z,\wtt)-G(t,z_*,\wtt_*), \label{diff-wtt-ev}
\end{align}
and that
\begin{equation}\label{diff-ev-idata}
(z-\zb,\wtt-\wttb)|_{t=0} = (z_*-z_*,\wtt_*-\wtt_*)=(0,0).
\end{equation}
Expressing the right hand side of \eqref{diff-z-ev} as
\begin{align*}
F(t,z,\wtt)-F(t,z_*,\wtt_*)=& \bigl(\Btt_{11}(t,\wtt_*)-\Btt_{11}(t,\wtt)\bigr)\del{x}z + \Btt_{11}(t,\wtt_*)\del{x}(z_*-z)\\
&+ t^{1-3\sspeed}\bigl(\Btt_{12}(t,\wtt_*)-\Btt_{12}(t,\wtt)\bigr)\del{x}w + t^{1-3\sspeed}\Btt_{12}(t,\wtt_*)\del{x}(\wtt_*-\wtt),
\end{align*}
we can, using similar arguments as above, employ the calculus inequalities from the appendix in conjunction with Corollary \ref{cor-Ac0Bc} and Lemma \ref{Btt-diff-lem} to estimate $F(t,z,\wtt)-F(t,z_*,\wtt_*)$ by
\begin{align*}
\norm{F(t,z,\wtt)-F(t,z_*,\wtt_*)}_{H^{k-\ell-2}}
\leq C(t)\Bigl(& t^{(k-\ell)(1-3\sspeed)}\norm{\wtt(t)-\wtt_*}_{H^{k-\ell-1}} \\
&+ t^{(k-\ell-2)(1-3\sspeed)}\norm{z(t)-z_*}_{H^{k-\ell-1}} \Bigr)
\end{align*}
where
\begin{equation*}
 C(t) = C_0\bigl(\norm{\wtt(t)}_{H^{k-\ell-1}},\norm{z(t)}_{H^{k-\ell-1}},\norm{\wtt_*}_{H^{k-\ell-1}},\norm{z_*}_{H^{k-\ell-1}}\bigr). 
\end{equation*}
Together, this bound and the decay estimates from Remark \ref{rem-exist}.(v) imply that
\begin{equation}
\norm{F(t,z,\wtt)-F(t,z_*,\wtt_*)}_{H^{k-\ell-2}}
\lesssim t^{(2(k-\ell)-1)(1-3\sspeed)+1}, \quad 0<t\leq 1. \label{F-diff-bnd}
\end{equation}
Using similar arguments to estimate the right hand side of \eqref{diff-wtt-ev}, we find that
\begin{equation}
\norm{G(t,z,\wtt)-G(t,z_*,\wtt_*)}_{H^{k-\ell-2}}
\lesssim t^{(2(k-\ell)-2)(1-3\sspeed)+1}, \quad 0<t\leq 1.\label{G-diff-bnd}
\end{equation}

Now, because $\frac{1}{3}<\sspeed < \frac{k+1}{3k}$ and $k-\ell \geq 3$, we have
\begin{equation*}
(2(k-\ell)-1)(1-3\sspeed)+1 > -1 + \frac{2 \ell +1}{k}>-1.
\end{equation*}
Due to this inequality and the estimates \eqref{F-diff-bnd}-\eqref{G-diff-bnd}, we can integrate \eqref{diff-z-ev}-\eqref{diff-wtt-ev} in time from $0$ to $t$ to conclude, with the help of the initial condition \eqref{diff-ev-idata}, that differences $z-\zb$ and $\wtt-\wttb$ can be expressed as 
\begin{align}
z(t,x) - \zb(t,x) &= \int_0^t  F(s,z(s,x),\wtt(s,x))- F(s,z(x),\wtt_*(x))\, ds, \label{z-diff-sol} \\
\wtt(t,x) - \wttb(t,x) &= \int_0^t G(s,z(s,x),\wtt(s,x))- G(s,z(x),\wtt_*(x))\, ds,  \label{wtt-diff-sol}
\end{align}
for $(t,x)\in [0,1]\times \Tbb$, and bounded by
\begin{align}
\norm{z(t)-\zb(t)}_{C^{k-\ell-3,\frac{1}{2}}}\lesssim \norm{z(t)-\zb(t)}_{H^{k-\ell-2}} \lesssim&  t^{(2(k-\ell)-1)(1-3\sspeed)+2}, \label{diff-z-bnd}\\
\norm{\wtt(t)-\wttb(t)}_{C^{k-\ell-3,\frac{1}{2}}}\lesssim \norm{\wtt(t)-\wttb(t)}_{H^{k-\ell-2}} \lesssim&  t^{(2(k-\ell)-2)(1-3\sspeed)+2},\label{diff-wtt-bnd}
\end{align}
for all $t\in [0,1]$. 

\begin{rem} \label{rem-refined-approx} Due to the inequalities 
\begin{gather*}
(2(k-\ell)-1)(1-3\sspeed)+2-((k-\ell)(1-3\sspeed)+1)=(k-\ell-1)(1-3\sspeed)+1>0,
\intertext{and}
(2(k-\ell)-2)(1-3\sspeed)+2-( (k-\ell-1)(1-3\sspeed)+1)=(k-\ell-1)(1-3\sspeed)+1>0,
\end{gather*}
it follows from the decay estimates \eqref{diff-z-bnd}-\eqref{diff-wtt-bnd} that the refined approximation $(\zb,\wttb)$ determines a more accurate approximation of the full solution $(z,\wtt)$ near $t=0$ compared with using just the limit functions $(z_*,\wtt_*)$, c.f.~Remark \ref{rem-exist}.(v), albeit in a norm of one order less differentiability. 
\end{rem}

We now assume that $x_0\in \Tbb$ satisfies\footnote{Recall from Theorem \ref{thm-non-empty} that there exists open sets of initial data that yield asymptotic functions $\wtt_*$ that vanish somewhere on $\Tbb$.}
\begin{equation*}
\wtt_*(x_0)=0.
\end{equation*}
Then by \eqref{phi-def}, \eqref{B-def} and \eqref{AcBc-def}, we have that
\begin{equation*}
\Btt(t,\wtt_*(x_0))= \begin{pmatrix} 0 & -1\\
-\sspeed & 0\end{pmatrix}.
\end{equation*}
By \eqref{zb-sol}-\eqref{wttb-sol}, the definitions of $F$ and $G$, see \eqref{z-ev}-\eqref{wtt-ev} and the fact that $\frac{1}{3}<\sigma <\frac{k+1}{3k} < \frac{2}{3}$, it follows that $\zb(t,x_0)$ and $\wttb(t,x_0)$ are given by
\begin{align}
\zb(t,x_0) &= z_*(x_0) + \int_{0}^t s^{1-3\sspeed}\del{x}\wtt_*(x_0) \, ds = z_*(x_0) + \frac{1}{2-3\sspeed}t^{2-3\sspeed}\del{x}\wtt_*(x_0), \label{zb-asymp} \\
\wttb(t,x_0) &= \sspeed \int_0^t s^{3\sspeed-1}\del{x}z_*(x_0)\,ds = \frac{1}{3}t^{3\sspeed}\del{x}z_*(x_0). \label{wbtt-asymp}
\end{align} 
Further assuming that $\ell=0$, we set 
\begin{equation*}
\nu := (2k-2)(1-3\sspeed)+2-3\sspeed = (2k-1)(1-3\sspeed)+1,
\end{equation*}
which we note satisfies
\begin{equation*}
\nu < 2-3\sspeed,
\end{equation*}
since $2k-2>0$ and $1-3\sspeed < 0$ by assumption. 
Because $0<\frac{2k}{3(2k-1)}< \frac{k+1}{3k}$ 
for $k\geq 3$, we can ensure that
\begin{equation} \label{nu-range}
0<\nu<1 \AND 
\nu +1 -(2-3\sspeed) = \nu + 3\sspeed-1>0
\end{equation}
by restricting  
$\sigma$ so that it lies in the interval
\begin{equation*}
\frac{1}{3} < \sspeed < \frac{2k}{3(2k-1)}\leq \frac{2}{5}.
\end{equation*}
In this case, we deduce from \eqref{zb-asymp}-\eqref{wbtt-asymp} and the decay estimates \eqref{diff-z-bnd}-\eqref{diff-wtt-bnd} that
\begin{gather}
z(t,x_0) =  z_*(x_0) + \frac{1}{2-3\sspeed}t^{2-3\sspeed}\del{x}\wtt_*(x_0) + \Ord(t^{\nu+1}) \label{z-x0-asymp}
\intertext{and}
t^{1-3\sspeed}\wtt(t,x_0) = t\biggl( \frac{1}{3}\del{x}z_*(x_0) + \Ord(t^\nu)\biggr). \label{wtt-x0-asymp}
\end{gather}

Using Taylor's Theorem to expand $\phi^{-1}(w)$ about $w=0$, we find, with the help of \eqref{cov}, \eqref{phi-def} and \eqref{dphi},   that 
\begin{equation*}
\phi^{-1}(w) = \Bigl(\frac{d\phi}{du}\Bigl|_{u=0}\Bigr)^{-1}w +\Ord(w^2) =w+\Ord(w^2),
\end{equation*}
from which we get, by  \eqref{nu-range} and \eqref{wtt-x0-asymp}, that 
\begin{equation} \label{phi-inv-asymp}
\phi^{-1}\bigl(t^{1-3\sspeed}\wtt(t,x_0)\bigr) =t \biggl( \frac{1}{3}\del{x}z_*(x_0) + \Ord(t^\nu)\biggr).
\end{equation}
Similarly, using Taylor's Theorem to expand the derivative $(\phi^{-1})'(w)$ about $w=0$, we have
\begin{equation*}
(\phi^{-1})'(w) = \Bigl(\frac{d\phi}{du}\Bigl|_{u=0}\Bigr)^{-1} +\Ord(w) =1+\Ord(w),
\end{equation*}
and hence that
\begin{equation} \label{dphi-inv-asymp}
(\phi^{-1})'\bigl(t^{1-3\sspeed}\wtt(t,x_0)\bigr) =1 +  \Ord(t).
\end{equation}
We also note from the decay estimates from Theorem \ref{thm-exist}.(a) that 
\begin{align}
\del{x}z(t,x_0) &= \del{x}z_*(x_0) + \Ord\bigl(t^{\frac{1}{2}(\nu+2-3\sspeed)}\bigr)
\AND
\del{x}\wtt(t,x_0) =  \del{x}\wtt_*(x_0) + \Ord\bigl(t^{\frac{1}{2}(\nu +3 \sigma)}\bigr). \label{dwtt-asymp} 
\end{align}
Plugging \eqref{z-x0-asymp}, \eqref{phi-inv-asymp}, \eqref{dphi-inv-asymp} and \eqref{dwtt-asymp} into
\eqref{T2-sol} and \eqref{asymp-E} yields asymptotic expansions
\begin{align*}
t^{-3(1+\sspeed)}\rho(t,x_0)&= \rho_c e^{(1+\sspeed) z_*(x_0)} + \Ord(t^{2-3\sspeed}) , \\
t^{-1}\vt^i(t,x_0) &= -\delta^i_0+ \Ord(t),\\
\frac{\del{x}\rho(t,x_0)}{\rho(t,x_0)} &= (1+\sspeed)\del{x}z_*(x_0) + \Ord\bigl(t^{\frac{1}{2}(\nu+2-3\sspeed)}\bigr),
\end{align*}
which are valid for $0<t<t_0$ for some $t_0 \in (0,1]$ chosen sufficiently small. Letting $t\searrow 0$ in these expression yields
\begin{gather*}
\lim_{t\searrow 0} t^{-3(1+\sspeed)}\rho(t,x_0) = \rho_c e^{(1+\sspeed)z_*(x_0)}, \quad
\lim_{t\searrow 0} t^{-1}\vt^{i}(t,x_0) =  -\delta^i_0 \AND
\lim_{t\searrow 0} \frac{\del{x}\rho(t,x_0)}{\rho(t,x_0)} = (1+\sspeed)\del{x}z_*(x_0).
\end{gather*}

\section{Statements and Declarations}

\subsection*{Data availability statement}
There is no data associated with this article.

\subsection*{Competing interests statement}
No funds, grants, or other support was received in the production of this article. The author has no relevant financial or non-financial interests to disclose.

%

\appendix

\section{\label{calc}Calculus inequalities}
In this appendix, we collect, for the convenience of the reader, a number of calculus inequalities that we employ in this article. 
The proof of the following inequalities are well known, and in particular, the proofs of  
Theorems \ref{Sobolev}, \ref{GNII}, \ref{Product}.(i), \ref{Product}.(ii), and \ref{Moser} below can be found in the following references, respectively:  \cite[Thm.~4.12]{AdamsFournier:2003},   \cite[\S Lecture II]{Nirenberg:1959} or \cite[Thm.~10.1]{Friedman:1969},  \cite[Thm.~A.5]{Koch:1990}, \cite[\S VI.3]{Choquet_et_al:2000} and 
\cite[Cor.~6.4.5]{Hormander:1997}.

In the following theorems, $\alpha=(\alpha_1,\ldots,\alpha_{n})\in \Zbb_{\geq 0}^{n}$ denotes a multi-index and 
$D^\alpha = \del{1}^{\alpha_1}\del{2}^{\alpha_2}\cdots\del{n}^{\alpha_{n}}$ where the $\del{i}=\fdel{\;}{x^i}$, $1\leq i\leq n$, denote partial 
derivatives with respect to periodic coordinates $(x^i)=(x^1,x^2,\ldots,x^n)$ on the $n$-torus $\Tbb^n$. 

\begin{thm}{\emph{[H\"{o}lder's inequality]}} \label{Holder}
If $0< p,q,r \leq \infty$ satisfy $1/p+1/q = 1/r$, then
\begin{equation*}
\norm{uv}_{L^r} \leq \norm{u}_{L^p}\norm{v}_{L^q}
\end{equation*}
for all $u\in L^p(\Tbb^n)$ and $v\in L^q(\Tbb^n)$.
\end{thm}

\begin{thm}{\emph{[Sobolev's inequality]}} \label{Sobolev} Suppose
$1\leq p < \infty$, $s\in \Zbb_{> 0}$ and $k\in \Zbb_{\geq 0}$.
If $s>n/p>s-1$ and $0<\lambda \leq s-n/p$, then
\begin{equation*}
\norm{u}_{C^{k,\lambda}} \lesssim \norm{u}_{W^{s+k,p}}
\end{equation*}
for all $u\in W^{s+k,p}(\Tbb^{n})$. Furthermore, this inequality holds if \textit{(i)} $n=(s-1)p$, $p>1$ and  $0<\lambda < 1$, or \textit{(ii)} 
$n=s-1$, $p=1$, and $0<\lambda\leq 1$. 
\end{thm}

\begin{thm}{\emph{[Gagliardo-Nirenberg interpolation inequality]}}\label{GNII}
Suppose $s\in \Zbb_{>|\alpha|}$, $1\leq p,q,r \leq \infty$, and $a\in [0,1]$ satisfy 
\begin{equation*}
\frac{|\alpha|}{s}\leq a \leq 1 \AND  \frac{|\alpha|}{n}+a\biggl(\frac{1}{p}-\frac{s}{n}\biggr)+(1-a)\frac{1}{q} = \frac{1}{r}.
\end{equation*}
Then
\begin{equation*}
\norm{D^\alpha u}_{L^r} \lesssim \norm{u}^{1-a}_{L^q}\norm{u}^{a}_{W^{s,p}}
\end{equation*}
for all $u \in L^q(\Tbb^n)\cap W^{s,p}(\Tbb^n)$ except if $1<p<\infty$ and $s-|\alpha|-n/p$ is a nonnegative integer in which case the inequality holds
for $|\alpha|/s\leq a <1$.
\end{thm}

\begin{thm}{\emph{[Product and commutator estimates]}} \label{Product} $\;$

\begin{enumerate}[(i)]
\item
Suppose $1\leq p_1,p_2,q_1,q_2\leq \infty$, $s\in \Zbb_{\geq 1}$, $|\alpha|=s$ and
\begin{equation*}
\frac{1}{p_1}+\frac{1}{q_1} = \frac{1}{p_2} + \frac{1}{q_2} = \frac{1}{r}.
\end{equation*}
Then
\begin{align*}
\norm{D^\alpha (uv)}_{L^r} \lesssim \norm{u}_{W^{s,p_1}}\norm{v}_{L^{q_1}} + \norm{u}_{L^{p_2}}\norm{v}_{W^{s,q_2}} \label{clacpropB.2.1}
\intertext{and}
\norm{[D^\alpha,u]v}_{L^r} \lesssim \norm{D u}_{L^{p_1}}\norm{v}_{W^{s-1,q_1}} + \norm{D u}_{
W^{s-1,p_2}}\norm{v}_{L^{q_2}}
\end{align*}
for all $u,v \in C^\infty(\Tbb^{n})$.
\item[(ii)]  Suppose $s_1,s_2,s_3\in \Zbb_{\geq 0}$, $\;s_1,s_2\geq s_3$,  $1< p \leq \infty$, and $s_1+s_2-s_3 > n/p$. Then
\begin{equation*}
\norm{uv}_{W^{s_3,p}} \lesssim \norm{u}_{W^{s_1,p}}\norm{v}_{W^{s_2,p}}
\end{equation*}
for all $u\in W^{s_1,p}(\Tbb^{n})$ and $v\in W^{s_2,p}(\Tbb^{n})$.
\end{enumerate}
\end{thm}

\begin{thm}{\emph{[Moser's estimate]}}  \label{Moser}
Suppose  $1\leq p \leq \infty$, $s\in \Zbb_{\geq 1}$, $1\leq k\leq s$, $|\alpha|=k$ and $f\in C^s(U)$ where $U\subset \Rbb^N$ is open. Then
\begin{equation*}
\norm{D^\alpha f(u)}_{L^{p}} \lesssim \norm{D_u f}_{W^{s-1,\infty}(U)}(1+\norm{u}^{s-1}_{L^\infty})\norm{u}_{W^{s,p}}
\end{equation*}
for all $u \in L^\infty(\Tbb^{n},\Rbb^N)\cap W^{s,p}(\Tbb^{n},\Rbb^N)$ satisfying $u(\Tbb^n)\subset U$. Moreover, if $0\in U$, then
\begin{equation*}
\norm{f(u)-f(0)}_{L^p} \leq K\norm{u}_{L^{p}}
\end{equation*}
for all $u \in L^\infty(\Tbb^{n},\Rbb^N)$ satisfying $u(\Tbb^n)\subset U$ where $K$ is the Lipschitz constant of $f$ on $U$. 
\end{thm}

\bibliographystyle{amsplain}
\bibliography{FLRW_Kgtot_v9}

\end{document}